\documentclass[12pt]{article}
\usepackage{enumerate}
\usepackage{graphicx}
\usepackage[square]{natbib}
\usepackage{multirow}
\usepackage{amsmath}
\usepackage{mathrsfs}
\usepackage{amsfonts}
\usepackage{amssymb}
\newcommand{\qed}{\hfill  $\square$}
\usepackage{hypernat}

\newenvironment{proof}{\textbf{Proof: }}{$\hfill \square \\$}

\newtheorem{lemma}{Lemma}[section]

\newtheorem{theorem}{Theorem}[section]

\newtheorem{remark}{Remark}[section]
\newtheorem{condition}{Condition}
\newtheorem{step}{Step}
\newtheorem{example}{Example}

\newcommand{\E}{\mathbb E}

\newcommand{\Perp}{\perp \! \! \!  \perp}

\newcommand{\prob}{\mathbb P}

\newcommand{\eps}{\epsilon}

\newcommand{\R}{\mathbb{R}}

\newcommand{\bii}{\vec{i}}
\renewcommand{\j}{\vec{j}}
\def\T{{ \mathrm{\scriptscriptstyle T} }}

\usepackage[colorlinks,citecolor=blue,urlcolor=blue,linkcolor=red]{hyperref}

\begin{document}

\def\baselinestretch{1.2}

%\begin{frontmatter}

% "Title of the paper"
\title{\textsc{On Testing Independence and Goodness-of-fit in Linear Models}}
\author{\begin{tabular}{c}
Arnab Sen and Bodhisattva Sen\footnote{Supported by NSF Grants DMS-1150435 and AST-1107373}\\
University of Minnesota and Columbia University
\end{tabular}}
\date{}
\maketitle
\fontsize{12}{18pt}
\selectfont

%\begin{frontmatter}
%\title{On Goodness-of-fit and Independence Testing in Linear Regression}
%%Testing for the Independence of the Predictor and the Error variables in Linear Regression
%\runtitle{Test of Independence}
%
%
%\begin{aug}
%\author{\fnms{Arnab} \snm{Sen}\ead[label=e1]{arnab@math.umn.edu}}
%\author{and \fnms{Bodhisattva} \snm{Sen}\thanksref{t1} \ead[label=e2]{bodhi@stat.columbia.edu}}
%\thankstext{t1}{Supported by NSF Grants DMS-1150435 and AST-1107373.}
%\runauthor{Sen and Sen}
%
%\affiliation{University of Minnesota and Columbia University}
%
%\address{
%\printead{e1}\\
%\phantom{E-mail:\ }\printead*{e2}}
%
%%\address{Address of the Third author\\
%%Usually a few lines long\\
%%Usually a few lines long\\
%%\printead{e3}\\
%%\printead{u1}}
%\end{aug}

\begin{abstract}
\noindent We consider a linear regression model and propose an omnibus test to simultaneously check the assumption of independence between the error and the predictor variables, and the goodness-of-fit of the parametric model. Our approach is based on testing for independence between the residual obtained from the parametric fit and the predictor using the Hilbert--Schmidt independence criterion \citep{KerIndepNIPS08}. The proposed method requires no user-defined regularization, is simple to compute, based merely on pairwise distances between points in the sample, and is consistent against all alternatives. We develop distribution theory  for the proposed test statistic, both under the null and the alternative hypotheses, and devise a bootstrap scheme to approximate its null distribution. We prove the consistency of the bootstrap scheme. A simulation study shows that our method has better power than its main competitors.  Two real datasets are analyzed to demonstrate the scope and usefulness of our method.

 \end{abstract}


%\begin{keyword}[class=AMS]
%\kwd[Primary ]{60K35}
%\kwd{60K35}
%\kwd[; secondary ]{60K35}
%\end{keyword}
%
%\begin{keyword}
%\kwd{Bootstrap, goodness-of-fit test, linear regression, model checking, reproducing kernel Hilbert space, test of independence}
%\end{keyword}
%
%\end{frontmatter}

\section{Introduction}\label{sec:intro}
In regression analysis, given a random vector $(X,Y)$ where $X$ is a $d_0$-dimensional predictor and $Y$ is the one-dimensional response, we want to study the relationship between $Y$ and $X$. In the most general form, the relationship can always be summarized as
\begin{eqnarray}\label{eq:Model}
	Y = m(X) + \eta,
\end{eqnarray}
where $m$ is the regression function and $\eta = Y - m(X)$ is the error that has conditional mean $0$ given $X$. In linear regression, we assume that $m$ belongs to a parametric class, e.g.,
\begin{equation}\label{eq:LinModel}
\mathcal{M}_\beta = \{g(x)^\T \beta: \beta \in \R^d \},
\end{equation} 
where $g(x) = (g_1(x),\ldots, g_d(x))^\T$ is the vector of known predictor functions and $\beta$ is the finite-dimensional unknown parameter.  Moreover,  for the validity of the standard theory of inference in linear models, e.g., hypothesis testing and confidence intervals,  it is crucial that the error $\eta$ does not depend on the predictor $X$. Thus to validate the adequacy of a linear model it is important to have statistical tests that can check, given independent and identically distributed  data $(X_1,Y_1), \ldots, (X_n,Y_n)$ from the regression model~\eqref{eq:Model}, the above two assumptions, namely, the correct specification of the parametric regression model and the independence of $X$ and $\eta$. % It is also desirable that such tests do not require any smoothness assumption on the joint distribution of  $(X, \eta)$. 

%Indeed there is an enormous body of literature on the goodness-of-fit tests for  the linear model under the assumption that $X$ and $\eta$ are independent. 
Several tests for the goodness-of-fit of a parametric model have been proposed under different conditions on the distribution of the errors and its dependence on the predictors; see~\citet{CoxEtAl88}, \citet{Azzalini93}, \citet{Eubank90}, \citet{HM93}, \citet{Fan01}, \citet{Stute97}, \citet{Guerre05}, \citet{Christensen10} and the references therein. Most of these tests assume that the errors are homoscedastic and sometimes even normal. Also, any such test using a nonparametric regression estimator runs into the problem of choosing a number of tuning parameters, e.g., smoothing bandwidths.

%\cite{CoxEtAl88} introduced tests of the null hypothesis that a regression function has a particular parametric structure under the assumption of independent homoscedastic normal errors. \cite{Azzalini93} used nonparametric regression to check linear relationships under independent homoscedastic errors; also see \cite{Eubank90} and \cite{Hardle93}. In \cite{Fan01}, some tests are proposed for examining the adequacy of a family of parametric models against large nonparametric alternatives under the assumption of independence and normality of the errors. In \cite{Guerre05} the authors propose data-driven smooth tests for a parametric regression function. For other relevant work on this topic see \cite{Christensen10}, \cite{Xia09} and the references therein. Any test using a nonparametric regression estimator runs into an ill-posed problem requiring the delicate choice of a number of tuning parameters, e.g., smoothing parameter(s). However, a few alternative approaches have been developed that circumvents these problems; see e.g. \cite{Bierens90}, \cite{Stute97}. However, most of these tests are difficult to implement and do not usually work well when the dimension of $X$ is not low.

%However, the normality assumption of the errors can be replaced by any mean zero distribution with finite second moment, for moderate sample sizes.  Thus the correct specification of the parametric regression component and the independence of $X$ and $\eta$ are the two most crucial assumptions. 

Very few methods in the literature test the independence of the predictor and the error. There has been much work on testing for homoscedasticity of the errors; see, for example, \citet{CW83}, \citet{BP79}, \citet{K08} and the references therein. However,  the  dependence between $X$ and $\eta$ can go well beyond simple heteroscedasticity. In the nonparametric setup, \citet{Einmahl08B, Einmahl08} propose tests for independence but only for univariate predictors. Generalization to the multivariate case is recently considered in \citet{Neumeyer10}; also see \citet{Neumeyer09}. %All these nonparametric tests do require some degree of smoothness on the distribution of $(X, \eta)$. 

It can be difficult to test the goodness-of-fit of the parametric model and the independence of $\eta$ and $X$ separately as they often have confounding effects. Any procedure testing for the independence of $\eta$ and $X$ must assume that the model is correctly specified as $\eta$ can only be reliably estimated under this assumption. On the other hand,  many  goodness-of-fit tests crucially use the independence of $\eta$ and $X$. In this paper we propose an omnibus easy-to-implement test to simultaneously check the assumption of independence of $X$ and $\eta$, denoted by $X \Perp \eta$, and the goodness-of-fit of the linear regression model, i.e., test the null hypothesis
\begin{equation}\label{eq:H_0}
H_0: X \Perp \eta,  \ \ m \in \mathcal{M}_\beta.
\end{equation} 
%which works without {\em any} smoothness condition on the distribution of  $(X, \eta)$ (e.g., existence of density). 
%Thus in one simple procedure we test whether the usual inferential techniques in the linear model are valid. 
Even when we consider the predictor variables fixed, our procedure can be used to check whether the conditional distribution of $\eta$ given $X$ depends on $X$. This will, in particular, help us detect heteroscedasticity. As far as we are aware, no test can simultaneously check for these two crucial model assumptions in linear regression.

Our procedure is based on testing for the independence of $X$ and the residual, obtained from fitting the parametric model,  using the Hilbert--Schmidt independence criterion \citep{KerIndepNIPS08}. Among the virtues of this test is that it is automated, that is,  requires no user-defined regularization, extremely simple to compute, based merely on the distances between points in the sample, and is consistent against all alternatives. Also, compared to other measures of dependence, the Hilbert--Schmidt independence criterion does not require smoothness assumption on the joint distribution of $X$ and $\eta$, e.g., existence of a density, and its implementation is not computationally intensive when $d_0$ is large. Moreover, this independence testing procedure also yields a novel approach to testing for the goodness-of-fit of the fitted regression model: under model mis-specification, the residuals, although uncorrelated with the predictors by definition of the least squares procedure, are very much dependent on the predictors, and the Hilbert--Schmidt independence criterion can detect this dependence; under $H_0$, the test statistic exhibits $n^{-1}$-rate of convergence, whereas, under dependence, we observe $n^{-1/2}$-rate of convergence for the centered test statistic.

We find the limiting distribution of the test statistic, under both the null and alternative hypotheses. Interestingly, the asymptotic distribution is very different from what would have been obtained if the true error $\eta$ were observed. To approximate the null distribution of the test statistic, we propose a bootstrap scheme and prove its consistency. The usual permutation test, which is used quite often in testing independence, cannot be directly used in this scenario as we do not observe $\eta$.

The paper is organized as follows: in Section~\ref{TestInd} we introduce the HSIC and discuss other measures of dependence. We formulate the problem and state our main results in Section~\ref{TestS}. A bootstrap procedure to approximate the distribution of the test statistic is developed in Section~\ref{Boots}. A finite sample study of our method along with some well-known competing procedures is presented in Section~\ref{Simul}. In Section~\ref{AppA}, Appendix A, we present a result on triangular arrays of random variables that will help us understand the limiting behavior of our test statistic under the null and alternative hypotheses, and yield the consistency of our bootstrap approach. The
proofs of the main results are given in Section~\ref{AppB}, Appendix B.

\section{Testing independence of two random vectors}\label{TestInd}
We briefly review the Hilbert--Schmidt independence criterion for testing the independence of two random vectors; see \citet{KerIndepALT05, KerIndepNIPS08} and \citet{EquivRKHS12}. We start with some background and notation. By a reproducing kernel Hilbert space  $\mathcal{F}$ of functions on a domain $\mathcal{U}$ with a positive definite kernel $k: \mathcal{U} \times \mathcal{U} \to \R$ we mean a Hilbert space of functions from $\mathcal{U}$ to $\R$ with inner product $\langle \cdot, \cdot \rangle$, satisfying the reproducing property
\begin{eqnarray*}
	\langle f, k(u,\cdot) \rangle & = & f(u), \qquad (f \in \mathcal{F}; u \in \mathcal{U}).
\end{eqnarray*}
 We say that $\mathcal{F}$ is  characteristic if and only if the map $$ P \mapsto \int_\mathcal{U} k(\cdot, u) d P(u),$$ is injective on the space of all Borel probability measures on $\mathcal{U}$ for which $\int_{\mathcal{U}} k(u,u) dP(u) < \infty$.  Likewise, let $\mathcal{G}$ be a second reproducing kernel Hilbert space on a domain $\mathcal{V}$ with positive definite kernel $l$. Let $P_{uv}$ be a Borel probability measure defined on $\mathcal{U} \times \mathcal{V}$, and let $P_u$ and $P_v$ denote the respective marginal distributions on $\mathcal{U}$ and $\mathcal{V}$. Let $(U,V) \sim P_{uv}$. Assuming that 
\begin{eqnarray}\label{eq:KernExp}
 \E\{k(U,U)\},
 \E\{l(V,V)\} <\infty,
\end{eqnarray} 
the Hilbert--Schmidt independence criterion of $P_{uv}$ is defined as 
\begin{eqnarray}\label{eq:theta}
	\theta(U, V) & = & \E \{k(U,U') l(V,V')\} +  \E\{k(U,U')\} \; \E \{ l(V,V')\}  -  2 \E \{k(U,U') l(V,V'')\},\;\;\;
\end{eqnarray}
where $(U', V'), (U'', V'')$ are independent and identically distributed copies of $(U,V)$. It is not hard to see that  $\theta(U, V) \ge 0$. More importantly, when $\mathcal{F}$ and $\mathcal{G}$ are  characteristic \citep{Lyons11,EquivRKHS12}, and \eqref{eq:KernExp} holds, then 
\[ \theta(U, V) = 0 \quad \text{ if and only if  } \quad P_{uv} = P_u \times P_v.\]

Given an independent and identically distributed sample  $(U_1,V_1), \ldots, (U_n,V_n)$ from $P_{uv}$, we want to test whether $P_{uv}$ factorizes as $P_u \times P_v$. For the purpose of testing independence, we will use a biased but computationally simpler empirical estimate of $\theta$ \citep[Definition 2]{KerIndepALT05}, obtained by replacing the unbiased $U$-statistics with the $V$-statistic
\begin{equation}\label{eq:TestV}
	\hat \theta_n =  \frac{1}{n^2} \sum_{i, j}^n k_{ij} l_{ij} + \frac{1}{n^4} \sum_{i, j, q, r}^n k_{ij} l_{qr} - 2 \frac{1}{n^3} \sum_{i, j, q}^n k_{ij} l_{iq} = \frac{1}{n^2} \mbox{trace}(KHLH),
\end{equation} 
where the summation indices denote all $t$-tuples drawn with replacement from $\{1, \ldots, n\}$, $t$ being the number of indices below the sum, $k_{ij}= k(U_i,U_j)$, and $l_{ij} = l(V_i, V_j)$, $K$ and $L$ are $n \times n$ matrix with entries $k_{ij}$  and $l_{ij}$, respectively, $H = I - n^{-1}11^\T$, and $1$ is the $n \times 1$ vector of ones. The cost of computing this statistic is $O(n^2)$; see \cite{KerIndepALT05}.

Examples of translation invariant characteristic kernel functions on $\R^p$, for $p\ge 1$, include the Gaussian radial basis function kernel $k(u, u') = \exp (-\sigma^{-2} \|u  - u'\|^2)$, $\sigma >0$, the Laplace kernel $k(u, u') = \exp (-\sigma^{-1} \|u  - u'\|)$, the inverse multiquadratics $k(u, u') = (\beta + \|u  - u'\|^2)^{-\alpha}$, $\alpha,\beta > 0$, etc. We will use the Gaussian kernel in our simulation studies and data analysis.

One can, in principle, use any other test of independence and develop a theory parallel to ours. The choice of the Hilbert--Schmidt independence criterion is motivated by a number of computational and theoretical advantages,  see, e.g., \citet{KerIndepALT05, KerIndepNIPS08}. The recently developed method of distance covariance, introduced by \cite{SzekelyRizzoBakirov07} and \cite{SzekelyRizzo09},  has received much attention in the statistical community. It tackles the problem of testing and measuring dependence between two random vectors in terms of a weighted $L_2$-distance between characteristic functions of the joint distribution of two random vectors and the product of their marginals; see \cite{EquivRKHS12} for a comparative study of the Hilbert--Schmidt independence criterion and the distance covariance methods. However, the kernel induced by the  semi-metric used in the distance covariance method \citep{EquivRKHS12} is not smooth and hence is difficult to study theoretically, at least using our techniques.

\section{Method}
\subsection{Test statistic}\label{TestS}
We consider the regression model \eqref{eq:Model}. We denote by $Z = (X, \eta) \sim P$ where $Z \in \R^{d_0} \times \R$ and $\E(\eta \mid X) = 0$. Let $P_X$ and $P_\eta$ be  the marginal distributions of $X$ and $\eta$ respectively. To start with, we will assume that $m$ does not necessarily belong to $\mathcal{M}_\beta$, as defined in \eqref{eq:LinModel}. Assuming that $ \E \{g(X) g(X)^\T\} < \infty$, $ \E \{ m(X)^2 \} <\infty$ and $\E (\eta^2) <\infty$, let us define  $D^2(\beta) = \E[\{Y - g(X)^\T \beta \}^2],$ for $\beta \in \R^d$. From the definition of $m$,  $D^2(\beta) = \E[ \{Y - m(X) \}^2] + \E[ \{m(X) - g(X)^\T \beta\}^2]$. 
%\begin{equation*}\label{eq:Dist2}
%	D^2(\beta) := \E[ (Y - m(X) )^2] + \E[ (m(X) - g(X)^\T \beta)^2]. 
%\end{equation*} 
The function $D^2$ is minimized at $\tilde \beta_0$ if and only if $\tilde \beta_0$ is a minimizer of $\tilde D^2(\beta) = \E[ \{m(X) - g(X)^\T \beta\}^2]. $ 
The quantity $\tilde D^2(\tilde \beta_0)$ measures the distance between the true $m$ and the hypothetical model $\mathcal{M}_\beta$. Clearly, if $m(X) = g(X)^\T \beta_0$, then $\beta_0 =\tilde \beta_0$. Under the assumption that $\E \{g(X) g(X)^\T\}$ is invertible,  $\tilde D^2( \beta)$ has the unique minimizer $$ \tilde \beta_0 = \E \{g(X) g(X)^\T\} ^{-1} \E \{m(X) g(X)\}.$$ 
Thus, $ g(x)^\T \tilde \beta_0$ is the closest function, in the least squares sense,  to $m(x)$ in $\mathcal{M}_\beta$. 

Given independent and identically distributed  data $(X_1,Y_1),  \ldots, (X_n,Y_n)$ from the regression model \eqref{eq:Model}, we compute the least squares estimator  in the class $\mathcal{M}_\beta$ as
\begin{equation}\label{eq:betahat}
\hat \beta_n = \arg \min_{\beta \in \mathbb{R}^d} \sum_{i=1}^n \big\{Y_i - g(X_i)^\T \beta\big\}^2.
\end{equation}
Then the least squares estimator $\hat \beta_n$ is  $$\hat \beta_n = A_n^{-1} \Big \{n^{-1} \sum_{i=1}^n g(X_{i}) Y_{i} \Big \},  \quad  A_n = n^{-1}\sum_{i=1}^n g(X_{i}) g(X_{i})^\T,$$ provided that $A_n$ is invertible. Let
\begin{equation}\label{eq:Residual}
	e_i = Y_i - g(X_i)^\T \hat \beta_n  \quad (i=1,\ldots, n)
\end{equation}
 be the observed residuals. The test statistic we consider is 
\begin{eqnarray}\label{eq:TestS}
	T_n =  \frac{1}{n^2} \sum_{i, j}^n k_{ij} l_{ij} + \frac{1}{n^4} \sum_{i, j, q, r}^n k_{ij}  l_{qr} -  \frac{2}{n^3} \sum_{i, j, q}^n k_{ij} l_{iq},
\end{eqnarray} 
where $k_{ij} = k(X_i,X_j)$, and $l_{ij} = l(e_i,  e_j)$ with $k$ and $l$ being characteristic kernels defined on  $\R^{d_{0}} \times \R^{d_{0}}$ and $\R \times \R$ respectively. Our test statistic is almost identical  to the empirical estimate $\hat \theta_n$ of the Hilbert--Schmidt independence criterion between $X$ and $\eta$ described in \eqref{eq:TestV} except for the fact that we replace the unobserved errors $\eta_i$ by the observed residuals $e_i$. 

\subsection{Convergence of $T_n$ under null and alternative hypotheses}
For any $u = (u_1,\ldots, u_p) \in \R^p$, we define the $\ell_\infty$-norm of $u$ as $|u|_\infty = \max_{1 \le i \le p} |u_i|$. We will assume throughout the paper that 
 \begin{condition} \label{c:I}
 $A =\E \{g(X) g(X)^\T\}$ is invertible.  %This is [I]
\end{condition}
Moreover, we will always assume the following conditions  on the kernels $k, l$.
 \begin{condition} \label{c:K}
 The kernels $k$ and $l$ are characteristic; $k$ is continuous and $l$ is twice continuously differentiable. Denoting the partial derivatives of $l$ as $l_x (x, y) = \partial_x l(x, y), l_{xy}(x, y)  = \partial_x \partial_y l(x, y)$, etc., we assume that $l_{xx}$, $l_{xy}$ and $l_{yy}$ are Lipschitz continuous with Lipschitz constant $L$ with respect to  the $\ell_\infty$-norm.
% This is [K]  
\end{condition} 
We study the behavior of the test statistic $T_n$ under the null hypothesis \eqref{eq:H_0}, and also under the following different scenarios:
\begin{eqnarray}\label{eq:AltHyp}
	H_1:  X  \not \Perp \eta, m \in \mathcal{M}_\beta, \quad
	H_2:  X \Perp \eta, m \notin \mathcal{M}_\beta, \quad
	H_3:  X  \not \Perp \eta, m \notin \mathcal{M}_\beta.
\end{eqnarray} 
%This is (M.a), (M.b) and (M.c).
To find the limiting distribution of $T_n$ under $H_0$, we will assume the following set of moment conditions on $X$ and $\eta$:
\begin{condition} \label{c:M}
(a) $\E \{ |g(X) |_\infty^{2}\} < \infty$;  \ \ \  (b) $\E(\eta^{2}) < \infty$; 

	 (c)  $\E \left [  k^2(X_{q}, X_{r}) \{1 + |g(X_{s})|^2_\infty \} \{1+ |g(X_{t})|^2_\infty \} \right ] < \infty$, \ \ \  $(1 \le q, r, s, t \le 4)$;
	 
	(d)   $\E  \{  f^2(\eta_{q}, \eta_{r})  \} < \infty \text{ for } f = l, l_x, l_y, l_{xx}, l_{yy}, l_{xy},$ \ $(1 \le q, r\le 2)$.
\end{condition}
\begin{theorem} \label{thm:LimDist-T_n}
Suppose that Conditions~\ref{c:I}, \ref{c:K} and \ref{c:M} hold. Then, under $H_0$, $nT_n \to \chi$ in distribution,  where $\chi$ has  a non-degenerate distribution that depends on $P = P_X \times P_\eta$ and is  denoted  by $\chi = \chi(P_X \times P_\eta)$.
\end{theorem}
\begin{figure}
\begin{center}
\includegraphics[width=9cm,height=9cm]{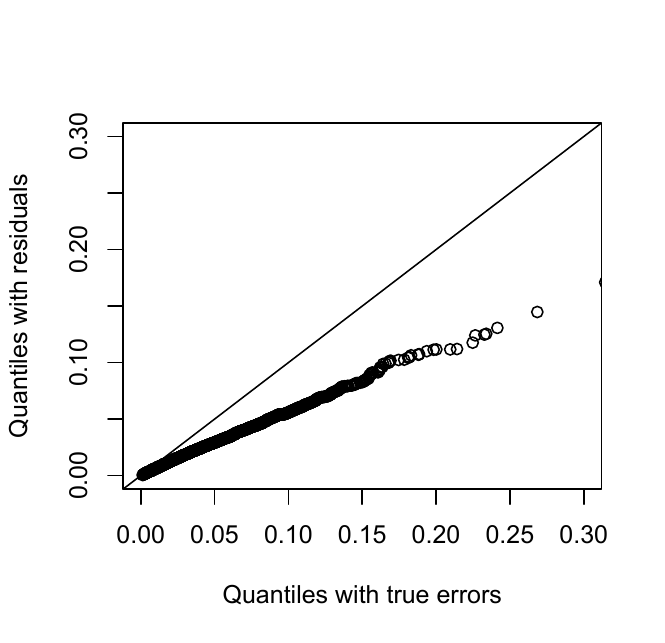}
\caption{Quantile-quantile plot of 5000 realizations of $n T_n$ obtained using the residuals versus 5000 realizations obtained using the true unknown errors in the linear model $Y = 1 + X + \eta$, where $\eta \sim N(0,\sigma^2 = 0.1)$, $X \sim N(0,1)$, $X \Perp \eta$, and $n =100$.}
\label{fig:HSICComp}
\end{center}
\end{figure}
\begin{remark}
Though one might be tempted to believe that replacing the unobserved true errors $\eta_i$ by the residuals $e_i$ should not alter the limiting distribution of the test statistic, this turns out to have an effect; see Figure 1 in the supplementary material.
\end{remark}

\begin{remark}
The random variable $\chi$ can be expressed as a quadratic function of a Gaussian field. This is in contrast with the limiting description of degenerate V-statistics, which would appear if $e_i$ were replaced by the true errors $\eta_i$,  where the limiting random variable can be described as a quadratic function of a family of independent Gaussian random variables. The explicit description of $\chi$ is slightly complicated and is described in \textsection~\ref{DescLD}; see \eqref{eq:LimDist}. However, from a practical point of view, such a description is of little use, since $P$ is unknown to the user. 
\end{remark}

%\begin{figure}
%\begin{center}
%%\figurebox{10pc}{13pc}{}[Rplot.pdf]
%\includegraphics[scale=0.7]{Rplot.pdf}
%\caption{Quantile-quantile plot of 5000 realizations of $n T_n$ obtained using the residuals versus 5000 realizations obtained using the true unknown errors in the linear model $Y = 1 + X + \eta$, where $\eta \sim N(0,\sigma^2 = 0.1)$, $X \sim N(0,1)$, $X \Perp \eta$, and $n =100$.}
%\label{fig:HSICComp}
%\end{center}
%\end{figure}

%%%%%
Next we study the limiting behavior of our test statistic $T_n$  under the different alternatives $H_1, H_2$ and $H_3$ in \eqref{eq:AltHyp}.  We first introduce the error under model mis-specification as 
\begin{equation}\label{eq:Defeps}
\eps= m(X) - g(X) ^\T \tilde \beta_0 + \eta.
\end{equation}
If $m \in \mathcal{M}_\beta$, then $\eps \equiv \eta$. We assume the following set of moment conditions for $H_1, H_2$ and $H_3$. 
\begin{condition}\label{c:M'} Let
	(a) $\E\{|g(X) |_\infty^{2}\} < \infty$ and $\E\{m(X)^2\} < \infty$; 
	
	 (b) $\E(\eta^2) < \infty$ and $\E \{ |g(X) |_\infty^{2} \eps^2\} < \infty$;
	
	(c)  for any $1 \le q, r, s, t \le 4$,
	\[\begin{split}
	 &\text{(i) } \E \left \{  k^2(X_{q}, X_{r})l^2(\eps_s, \eps_t  ) \right \} < \infty,\\
	 &\text{(ii) } \E \left [  | k(X_{q}, X_{r})|   |\nabla l(\eps_s, \eps_t)|_\infty \big  \{|g(X_s)|_\infty +|g(X_t)|_\infty \big\} \right ] < \infty,\\
	 &\text{(iii) }  \E \left [  | k(X_{q}, X_{r})|   \{ |g(X_s)|^3_\infty +|g(X_t)|^3_\infty \} \right ] < \infty, \\
	 &\text{(iv) } \E \left [ | k(X_{q}, X_{r})|   |\mathrm{Hess}(l)(\eps_s, \eps_t)|_\infty \big\{ |g(X_{s})|^2_\infty + |g(X_{t})|^2_\infty\big\} \right]< \infty,
	 \end{split}\]
	 where  $ |\nabla l(\eps_s, \eps_t)|_\infty  =  \max \big \{| l_x(\eps_s, \eps_t  )|,  | l_y(\eps_s, \eps_t  )| \big\} $ and  $  |\mathrm{Hess}(l)(\eps_s, \eps_t)|_\infty$    $  =  \max \big\{ |l_{xx}(\eps_s, \eps_t)|, $     $  |l_{xy}(\eps_s, \eps_t)|,  |l_{yy}(\eps_s, \eps_t)| \big\}$. Here $(X_1,\eps_1), \ldots, (X_4,\eps_4)$ are independent and identically distributed copies of  $(X,\eps)$.
\end{condition}
%This is M'

\begin{theorem}\label{thm:LimDistAlt-T_n}
Suppose that Conditions~\ref{c:I}, \ref{c:K}  and \ref{c:M'} hold.  Assume further under $H_2$  that $m(X) - g(X) ^\T \tilde \beta_0$ is non-constant.
Then under $H_1, H_2$ or $H_3$, 
\begin{equation}\label{eq:LimDistAlt}
n^{1/2}(T_n - \theta) \to N(0, \sigma^2),
\end{equation}
in distribution, where  $\theta = \theta(X,\eps)$ is defined in \eqref{eq:theta} and  $\theta>0$. The variance $\sigma^2$ depends on the joint distribution of $(X, \eps)$; an expression for it  can be found in \eqref{eq:sigma_def}. 
\end{theorem}
\begin{remark} The parameters $\theta$ and $\sigma^2$ appearing in~\eqref{eq:LimDistAlt} depend on the joint distribution of $(X,\eps)$, and thus can be different under the three alternative hypotheses $H_1, H_2$ or $H_3$.
\end{remark}
\begin{remark}
It is of interest to investigate whether Theorems~\ref{thm:LimDist-T_n} and~\ref{thm:LimDistAlt-T_n} can be generalized to the case where $\mathcal{M}_\beta = \{m(\cdot,\beta): \beta \in \R^{d}\}$ is any smooth parametric family, and not necessarily linear as defined in~\eqref{eq:LinModel}. Our proof technique cannot be directly applied in this general framework as in this situation there is no closed form expression for $\hat \beta_n$ which complicates the theoretical analysis. However, we believe that, with assumptions analogous to those in \citet[pages 617--618]{Stute97}, our results  can be extended to general parametric models.  
\end{remark}

\subsection{Description of the limiting distributions}\label{DescLD}
To give an explicit description of $\chi(P_X \times P_\eta)$ appearing in Theorem~\ref{thm:LimDist-T_n} we need some notation, which we introduce below. Set $\zeta_n =  n^{-1/2} A ^{-1} \sum_{i=1}^n  g(X_{i}) \eta_{i}$ and define the quantities $k_{ij}$, $l_{ij}^{(p)}$ for  $p \in \{0,1,2\}$, as
\begin{eqnarray*}
k_{ij} = k(X_{i}, X_{j}), \ \   l_{ij}^{(0)}  = l_{ij} =  l(\eta_{i}, \eta_{j}),   \ \   l_{ij}^{(1)} = - \Big  \{l_x(\eta_{i}, \eta_{j}) g(X_{i}) + l_y(\eta_{i}, \eta_{j}) g(X_{j})\Big\}  \in \R^d,  \\
l_{ij}^{(2)} =   \Big \{ l_{xx}(\eta_{i}, \eta_{j})  g(X_{i}) g(X_{i})^\T+ l_{yy}(\eta_{i}, \eta_{j})  g(X_{j}) g(X_{j})^\T  +    2l_{xy}(\eta_{i}, \eta_{j})  g(X_{i}) g(X_{j})^\T \Big \} \in \R^{d \times d}.
\end{eqnarray*}
For $p \in \{0, 1, 2\}$, let $h^{(p)}$ be the   symmetric kernel  
\begin{equation}\label{def:hp}
h^{(p)}(Z_{i}, Z_{j},  Z_{q}, Z_{r}) = \frac{1}{4!} \sum_{(t, u, v, w)}^{(i, j, q, r)}  k_{tu} l_{tu}^{(p)} + k_{tu} l_{vw}^{(p)}  - 2 k_{tu}l_{tv}^{(p)}, 
\end{equation}
where $Z_i = (X_i,\eta_i)$ and the sum is  taken over all $4!$ permutations of $(i, j, q, r)$. We need the appropriate  projections of the symmetric kernel $h^{(p)}$.  Let us define 
\begin{eqnarray*}
h^{(0)}_2(z_1, z_2) =  \E  \{h^{(0)}(z_1, z_2, Z_3, Z_4) \},  \quad h_1^{(1)}(z_1) =  \E \{ h^{(1)}(z_1, Z_2, Z_3, Z_4)\}, \\ 
\Lambda  = \E \{ h^{(2)}(Z_1, Z_2, Z_3, Z_4)\}. \hspace{1.5in}
\end{eqnarray*}
The symmetric function $h^{(0)}_2$ admits a spectral decomposition $ h^{(0)}_2(z_1, z_2) = \sum_{r = 0}^\infty \lambda_r \varphi_r(z_1) \varphi_r(z_2)$
where $(\varphi_r)_{ r \ge 0}$ is an orthonormal basis of $L_2(\mathbb R^{d_0+1}, P)$. Since $h^{(0)}_2$ is degenerate of order $1$,  $\lambda_0 = 0,  \varphi_0 \equiv 1$.   Therefore, $\E\{\varphi_r(Z_1)\} = 0$ for each $r \ge 1$. Also, $\sum_{ r} \lambda_r^2  = \E \{ h^{(0)}_2(Z_1, Z_2)^2\} < \infty$. Define jointly Gaussian random variables
 \[ \mathcal{Z} = \{\mathcal{Z}_r\}_{ r \ge 1}, \quad \mathcal{N}= \{\mathcal{N}_i\}_{1 \le i \le d}, \quad \mathcal{W} = \{\mathcal{W}_i\}_{1 \le i \le d},\]
where $\mathcal{Z}_r$ are independent and identically distributed $N(0,1)$,  $\mathcal{N} \sim N_d(0, \Xi)$ and $\mathcal{W} \sim N_d(0,\sigma^2_0I)$, with $\sigma^2_0  = \E (\eta_1^2)$ and $ \Xi = \E\{h_1^{(1)}(Z_1) h_1^{(1)}(Z_1)^\T\}$.  Also, the covariance structure between the random variables  $\mathcal{Z}_r, \mathcal{N} $ and $\mathcal{W}$ is given by 
\begin{eqnarray*}
\E (\mathcal{Z}_r \mathcal{N} ) = \E \{\varphi_r(Z_1)h_1^{(1)}(Z_1)\},   \quad \E ( \mathcal{Z}_r  \mathcal{W})  =  A^{-1}\E \{g(X_1) \eta_1 \varphi_r(Z_1)\},\\
 \E (\mathcal{W} \mathcal{N}^\T  )  =  A^{-1}\E\{\eta_1 g(X_1)  h_1^{(1)}(Z_1)^\T\}. \hspace{1.0in}
\end{eqnarray*}
The limiting distribution $\chi$ is  the following quadratic function of the above Gaussian field, 
\begin{equation}\label{eq:LimDist}
 \chi(P_X \times P_\eta) = \sum_{ r=1}^\infty \lambda_r \mathcal{Z}_r^2+ \sum_{ i=1} ^d \mathcal{W}_i  \mathcal{N}_i + \frac12  \sum_{ i, j=1} ^d   \Lambda_{ij} \mathcal{W}_i \mathcal{W}_j,
\end{equation}
where $\Lambda_{ij}$ is the $(i,j)$-th entry of the matrix $\Lambda$.

Now we describe the parameters in \eqref{eq:LimDistAlt} appearing in Theorem~\ref{thm:LimDistAlt-T_n}. Note that $\theta$ is defined in \eqref{eq:theta}. Define $h^{(p)}(W_q, W_r, W_s, W_t)$ analogously as in \eqref{def:hp} where $W_1, \ldots, W_4$ are independent and identically distributed copies of $W = (X, \eps)$, and $\eps$ is defined in \eqref{eq:Defeps}. Set $h^{(0)}_1(w) =  \E \{h^{(0)}(w , W_{2},  W_{3}, W_{4})\} -  \theta(X, \eps)$ and  $\gamma  = \E \{ h^{(1)}(W_{1}, W_{2},  W_{3}, W_{4})\}$. Then
 \begin{equation}\label{eq:sigma_def}
\sigma^2 = \mathrm{var}\{h^{(0)}_1(W_1)+ \gamma^\T A^{-1} g(X_1)\eps_1\}. 
\end{equation}

\section{Consistency of the bootstrap}\label{Boots}
Theorem~\ref{thm:LimDist-T_n} is not very useful in computing the critical value of the test statistic $n T_n$, as the asymptotic distribution $\chi$ involves infinitely many nuisance parameters. An obvious alternative is use of resampling to approximate the critical value of the test. In independence testing problems, a natural choice is a permutation test; see e.g. \citet{SzekelyRizzo09}, \citet{KerIndepNIPS08}.  

However, as we are using the residuals $e_i$ instead of the true unknown errors $\eta_i$ in our test statistic, a permutation-based test  will not work.  Indeed, under the null hypothesis, the joint  distribution of $\{X_i, \eta_{\pi(i)}\}_{ 1 \le i \le n}$ remains unchanged under any permutation $\pi$ of $\{1,\ldots, n\}$, but that of $\{X_i, e_{\pi(i)}\}_{ 1 \le i \le n}$ is not invariant under $\pi$.

%-- fixing the $X_i$'s and reshuffling the $e_i$'s and recomputing the statistic -- will not work [{\bf Do we need to provide a more concrete reason?}].

In this section we show that the bootstrap can be used to consistently approximate the distribution of $n T_n$, under $H_0$. In the following we describe our bootstrap procedure. 

\begin{step}
	Let $\mathbb{P}_{n,e^o}$ be the empirical distribution of centered residuals, i.e., $$e_i^o = e_i - \bar{e} \quad (i=1,\ldots, n), $$ where $e_i$ is defined in~\eqref{eq:Residual} and $\bar{e} = n^{-1} \sum_{i=1}^n e_i$.  
\end{step}
\begin{step}
	 Generate an  independent and identically distributed bootstrap sample $\{X_{in}^*, \eta_{in}^*\}_{ 1 \le i \le n}$ of size $n$ from the measure $P_n = \mathbb{P}_{n,X} \times \mathbb{P}_{n,e^o}$ where $\mathbb{P}_{n,X}$ is the empirical distribution of the observed $X_i$'s.
\end{step}
\begin{step}
	 Define $$Y_{in}^* = g(X_{in}^*)^\T \hat \beta_n + \eta_{in}^* \quad (i = 1,\ldots, n),$$  where $\hat \beta_n$ is the least squares estimator obtained in \eqref{eq:betahat}. Compute the bootstrapped least squares estimator $\hat \beta_n^*$  using the bootstrap sample $(Y_{in}^*,X_{in}^*)$. Also compute the bootstrap residuals $$e_{in}^* = Y_{in}^* - g(X_{in}^*)^\T \hat \beta_n^*\quad  (i=1,\ldots, n).$$ 
\end{step}
\begin{step}
	 Compute the bootstrap test statistic $T_n^*$, defined as in \eqref{eq:TestS}, with $X_i$ replaced by $X_{in}^*$, and $e_i$ replaced by $e_{in}^*$, for $i = 1,\ldots, n$. We approximate the distribution of $n T_n$ by the conditional distribution of $n T_n^*$, given the data.
\end{step}

%\textbf{not needed [}
%Before we state the main result of this section we give a brief review of the bootstrap and its consistency. Given a sample ${\mathbf W}_n=\{W_1, W_2, \ldots,$ $ W_n\}\stackrel{\rm iid}{\sim} L$ from an unknown distribution $L$, suppose that the distribution function $H_n$ of some random variable $R_n \equiv R_n(\mathbf{W}_n, L)$ is of interest; $R_n$ is usually called a {\it root} and it can in general be any measurable function of the data and the distribution $L$. The bootstrap method can be broken into three simple steps:
%\begin{itemize}
%	\item[(i)] Construct an estimator $\hat{L}_n$ of $L$ from ${\mathbf W}_n$.
%	
%	\item[(ii)] Generate  ${\mathbf W}_n^{*} =\{W_1^{*},\ldots, W_{m_n}^{*}\} \stackrel{\rm iid} {\sim} \hat{L}_n$ given ${\mathbf W}_n$.
%	
%	\item[(iii)] Estimate $H_n$ by $\hat H_{n}$, the conditional CDF of $R_n({\mathbf W}_n^{*},\hat{L}_n)$
%    given ${\mathbf W}_n$.
%\end{itemize}
%Let $d$ denote the Prokhorov metric or any other metric metrizing weak convergence of probability measures. We say that $\hat H_{n}$ is {\it weakly consistent} if $d(H_n, \hat H_n)\stackrel{P}{\rightarrow} 0$; if $H_{n}$ has a weak limit $H$, this is equivalent to $\hat H_{n}$ converging weakly to $H$ in probability. Similarly, $\hat H_{n}$ is {\it strongly consistent} if $d(H_n, \hat H_n)\stackrel{a.s.}{\rightarrow} 0$.
%\textbf{ ] not needed}

Assume that we have an infinite array of random vectors $Z_1, Z_2, \ldots,$ where $Z_i = (X_i, \eta_i)$ are independent and identically distributed from $P$ defined on some  probability space $(\Omega,\mathcal{A}, \mathrm{pr})$. We denote by $\mathfrak{Z}$ the entire sequence $\left\{Z_i\right\}_{i \ge 1}$ and write $\mathrm{pr}_\omega=\mathrm{pr}(\cdot \mid \mathfrak{Z})$ and $\E_{\omega} =\E(\cdot \mid \mathfrak{Z})$ to denote conditional probability and conditional expectation, respectively, given $\mathfrak{Z}$. 

The following result shows that under $H_0$, the distribution of $n T_n^*$, given the data $\{X_i,Y_i\}_{ 1 \le i \le  n}$, almost surely,  converges to the same limiting distribution as that of $nT_n$. Thus the bootstrap procedure is strongly consistent and we can approximate the distribution function of $n T_n$ by $\prob_\omega(n T_n^* \le \cdot )$, and use it to find the one-sided cut-off for testing $H_0$. To prove the result, 
we will need similar but slightly stronger conditions than those stated in  Condition~\ref{c:M}.  Recall that $\eps = m(X) - g(X) ^\T \tilde \beta_0 + \eta$, and  set $\eps^o = \eps - \E (\eps)$.
\begin{condition} \label{c:M''}
There exists $\delta>0$ such that 

	(a) $\E \{|g(X) |_\infty^{4 + 2\delta} \} < \infty$ and $\E \{ |m(X)|^{2+\delta}\} < \infty$; \qquad (b) $\E \{ |\eta|^{2+\delta} \} < \infty$; 
	
(c)  $\E \left [  |k(X_{q}, X_{r})|^{2 + \delta} \{1 + |g(X_{s})|^{2 + \delta}_\infty \} \{1+ |g(X_{t})|^{2 + \delta}_\infty \} \right ] < \infty$ \quad $(1 \le q, r, s, t \le 4)$;

(d) for  $1 \le q, r\le 2$,  
 \begin{eqnarray*} 
	\E  \{  |l (\eps^o_{q}, \eps^o_{r})|^{2 + \delta} \} <  \infty,  \ \   \E [   \{1+|g(X_q) |^{2+\delta}_\infty \} |f (\eps^o_{q}, \eps^o_{r})|^{2 + \delta}  ]  <  \infty \quad   (f =  l_x, l_y),\\ 
	 \E [    \{1+|g(X_q) |^{2+\delta}_\infty + |g(X_q) |^{4+2\delta}_\infty \} |f (\eps^o_{q}, \eps^o_{r})|^{2 + \delta}  ] <  \infty  \quad    (f =   l_{xx}, l_{yy}, l_{xy}).
\end{eqnarray*}
\end{condition}

\begin{theorem}  \label{thm:BootsCons}
Suppose that Conditions~\ref{c:I}, \ref{c:K} and \ref{c:M''} hold. Then 
\begin{equation}\label{boot_conv}
n T_n^* \to \chi(P_X \
\times P_{\eps^o}),
\end{equation}
in distribution,  conditional on the observed data almost surely, where $\chi$ is described in~\textsection~\ref{DescLD}. As a consequence, under $H_0$, $n T_n^*$ converges to $\chi(P_X \times P_\eta)$ in distribution, conditional on the observed data almost surely.
\end{theorem}

\begin{remark}  It follows from  Theorem~\ref{thm:LimDistAlt-T_n} that  $nT_n  \to \infty $  in probability under $H_1, H_2$ or $H_3$. But by Theorem~\ref{thm:BootsCons},  the quantiles of  the conditional  distribution of $nT_n^*$ are tight.  Hence, the power of our test under $H_1, H_2$ or  $H_3$ converges to $1$ as $n \to \infty$.
\end{remark}

\begin{remark}
Since the limiting distribution $\chi(P_X \times P_{\epsilon^o})$ is a nontrivial quadratic function of certain  correlated  Gaussian random variables, it has a smooth density and  hence the convergence in \eqref{boot_conv}  implies the convergence of the $\alpha$-quantile, for any $\alpha \in (0,1)$. Therefore, using  the bootstrap distribution will yield an asymptotic level $\alpha$ test. 
\end{remark}

\begin{remark}
A natural choice for $k$ and $l$ is the Gaussian kernel.  In this case, we can take $k(u, u') = \exp(  -\sigma^{-2} \| u - u'\|^2)$ and $l(v, v') =  \exp(  - \gamma^{-2} | v - v'|^2)$ where $u, u' \in \R^{d_0}$, $v, v'\in \R$ and $\sigma$ and $\gamma$ are fixed parameters (can be taken to be $1$).  Then $k$ and $l$ satisfy Condition~\ref{c:K}. Since the Gaussian kernels are bounded with all their partial derivatives bounded,  Conditions~\ref{c:M}(d), \ref{c:M'}(c) and \ref{c:M''}(c)--(d)  are automatically satisfied for any joint distribution of $(X, \eta)$. Also, Condition~\ref{c:M}(c)  is  implied by the simpler condition $\E \{| g(X)|_\infty^4 \} < \infty$. 
\end{remark} 

\section{Simulation study and data analysis}\label{Simul}
\subsection{Models}

\begin{table}
\def~{\hphantom{0}}
\begin{center}
%\tbl{Percentage of times Models 1 and 2 were rejected when $\alpha = 0.05$}
\caption{Percentage of times Models 1 and 2 were rejected when $\alpha = 0.05$}
{%
\begin{tabular}{llccccccc}
 \\
%& \multicolumn{7}{c}{Model 1} & \multicolumn{7}{c}{Model 2} \\
& $\lambda$ & 0 & 5 & 10 & 15 & 20 & 25 & 50 \\[5pt]
$n$=100 & Model 1 & 4 & 16 & 26 & 31 & 34 & 40 & 41\\
				  & Model 2 & 5 & 20 & 29 & 32 & 35 & 35 & 36 \\[5pt]
$n$=200 & Model 1 & 5 & 38 & 66 & 74 & 80 & 83 & 90\\
               & Model 2 & 6 & 47 & 62 & 67 & 69 & 72 & 76
\end{tabular}}
\label{tablelabel}
%\begin{tabnote}
%U.S., United States of America; R, respondent.
%\end{tabnote}
\end{center}
\end{table}
In this section we investigate the finite sample performance of the proposed testing procedure based on $T_n$, as defined in \eqref{eq:TestS}, in two different scenarios: (a) testing for the independence of the error $\eta$ and the predictor $X$, as in \eqref{eq:Model}, when the regression model is well-specified; (b) testing for the goodness-of-fit of the parametric regression model when the independence of $\eta$ and $X$ is assumed. As discussed in~\textsection~\ref{sec:intro}, there are very few methods available to test (a), and hardly any when $d_0 > 2$. For the goodness-of-fit of the parametric regression model there has been quite a lot of work and we compare our procedure with six competing methods. 

\begin{table}
\def~{\hphantom{0}}
\begin{center}
\caption{Percentage of times Model 1 was rejected when $\alpha = 0.05$,  $n =100$ and $d_0 = 2,4,6$}
{%
\begin{tabular}{llcccccccccc}
 \\
& $a$ & 0 & 0.5 & 1 & 1.5 & 2 & 3 & 4 & 5 & 7 & 10 \\[5pt]
$p$=2 & $T_n$ & 6 & 7 & 8 & 14 & 21 & 43 & 69 & 89 & 99 & 100 \\
& $S_1$ & 5 & 6 & 7 & 8 & 13 & 22 & 37 & 54 & 82 & 97 \\
& $S_2$ & 5 & 6 & 7 & 11 & 19 & 37 & 63 & 83 & 98 & 100 \\
& $F$ & 8 & 9 & 9 & 9 & 10 & 12 & 17 & 23 & 50 & 92 \\
& $G$ & 4 & 6 & 5 & 6 & 5 & 7 & 9 & 20 & 55 & 93 \\
& $S_P$ & 5 & 6 & 5 & 5 & 6 & 7 & 11 & 21 & 57 & 91 \\
& $L$ & 10 & 12 & 13 & 16 & 22 & 33 & 48 & 65 & 88 & 99\\[5pt]
$p$=4 & $T_n$ & 4 & 4 & 6 & 7 & 10 & 21 & 35 & 55 & 88 & 100 \\
& $S_1$ & 6 & 5 & 6 & 6 & 6 & 10 & 12 & 15 & 31 &46 \\
& $S_2$ & 3 & 4 & 5 & 4 & 5 & 11 & 15 & 22 & 40 & 60 \\
& $F$ & 8 & 7 &  7 & 9 &  9 & 11 & 17 & 21 & 47 & 90 \\
& $G$ & 5 & 5 & 5 & 5 & 5 & 7 &  10 & 19 & 54 & 92 \\
& $S_P$ & 5 & 7 & 6 & 6 & 7 & 7 & 11 & 21 & 54 & 91 \\
& $L$ & 21 & 26 & 30 & 31 & 35 & 46 & 57 & 70 & 91 & 99 \\[5pt]
$p$=6 &  $T_n$ & 1 & 2 & 2 & 2 & 3 & 6 & 9 & 19 & 39 & 84 \\
&  $S_1$  & 5 &  6 &  5 & 6 &  5 & 6 & 9 & 9 & 9 & 17 \\ 
&  $S_2$  & 3 & 3 &  3 &  4 & 4 & 4 &  5 & 6 & 9 & 16  \\  
&  $F$ &  6 & 7 & 7 & 7 & 8 &  9 & 12 & 18 &  42 & 86 \\ 
&  $G$ & 5 & 5 & 5 & 5 & 5 & 6 & 10 & 19 & 53 & 92 \\ 
&  $S_P$ & 6 & 7 & 7 & 7 & 6 & 8 & 12 & 22 & 53 &  89 \\ 
&  $L$ & 34 & 41 & 43 & 48 & 45 & 53 & 64 & 77 & 89 & 98
\end{tabular}}
\label{tablelabel}
\end{center}
\end{table}

We consider two data generating models. Model 1 is adapted from \citet[Model 3]{StuteEtAl98} and can be expressed as
\begin{eqnarray}\label{eq:SimMdl1}
	Y = 2 + 5 X_1 -  X_2 + a X_1 X_2 + \eta, \nonumber
\end{eqnarray}
with predictor $X = (X_1, \ldots , X_{d_0})^\T$, where $X_1, \ldots, X_{d_0}$ are independent and identically distributed Uniform$(0,1)$, and $\eta$ is drawn from an independent normal distribution with mean 0. \cite{StuteEtAl98} used $d_0 = 2$ in their simulations but we use $d_0 = 2, 4, 6$. The other model, Model 2, is adapted from \citet[Example 4]{Fan01}  and can be written as
\begin{eqnarray}\label{eq:SimMdl2}
	Y = X_1 + a X_2^2 + 2 X_4 + \eta, \nonumber
\end{eqnarray}
where $X = (X_1, X_2, X_3, X_4)^\T$ is the predictor vector. The predictors $X_1, X_2, X_3$ are normally distributed with mean $0$ and variance $1$ and pairwise correlation $0.5$. The predictor $X_4$ is binary with success probability $0.4$ and independent of $X_1, X_2$ and $X_3$. Random samples of size $n$ are drawn from Model 1 and  Model 2 and a multiple linear regression model is fitted to the samples, without the $X_1 X_2$ and $X_2^2$ terms, respectively. Thus, these models are well-specified if and only if $a = 0$.

In all the following $p$-value calculations, whenever required, we use 1000 bootstrap samples to estimate the critical values of the tests. The rejection probabilities reported in all the tables are computed using 2000 independent replicates. To make our method invariant under linear transformations we work with standardized variables. To implement our method we take Gaussian kernels with unit bandwidths.
 
\subsection{Testing for the independence}
We consider the above two models with $a=0$ and 
\begin{eqnarray*}
	\eta \mid X_1 \sim N \left(0, \frac{1 + \lambda |X_1|}{2} \right),
\end{eqnarray*}
where $\lambda = 0, 5, 10, 15, 20, 25, 50$. Table 1 gives the percentage of times Model 1, with $d_0 = 4$, and Model 2 were rejected as the sample size $n$ and $\lambda$ vary, when $\alpha = 0.05$. As expected, the power of the test increases monotonically with an increase in $\lambda$ and $n$.

\subsection{Goodness-of-fit test for parametric regression} Under the assumption of  independence of $X$ and $\eta$, our procedure can be used to test the goodness-of-fit of the fitted parametric model. In our simulation study we compare the performance of our method with six other competing methods, which we describe below. 

\begin{table}
\def~{\hphantom{0}}
\begin{center}
\caption{Percentage of times Model 2 was rejected when $\alpha = 0.05$ and $n=100$}
{%
\begin{tabular}{lccccccccccc}
$a$ & 0 & 0.05 & 0.10 & 0.15 & 0.20 & 0.25 & 0.30 & 0.35 & 0.40 & 0.50 & 0.60 \\[5pt]
$T_n$ & 5 & 6 & 8 & 13 &19 & 34 & 43 & 56 &  66 &  84  & 91 \\
$S_1$ & 8 &  6 &  7 &  9 & 10 &  18 &  24 &  31 &  41 &  58 & 68  \\
$S_2$ & 7 &  6 & 8  & 11 & 13 &  22 & 30 & 37 & 42 &  57 &  69  \\ 
$F$ & 6 & 7 & 9 & 10 & 11 & 16 & 20 & 32 & 43 & 66 & 85 \\
$G$ & 4 & 5 & 7 &  7 &  6 &  8 & 12 & 16 &  25 & 36 & 50 \\
$S_P$ & 5 & 7 & 4 & 7 & 5 & 6 & 7 & 7 & 7 & 8 &  9 \\
$L$ &  10 &  8 & 10 & 9 & 14 & 16 & 22 & 27 &  32 & 49 &  58
\end{tabular}}
\label{tablelabel}
\end{center}
\end{table}

\cite{StuteEtAl98} used the empirical process of the regressors marked by the residuals to construct various omnibus goodness-of-fit tests. Wild bootstrap approximations were used to find the critical values of the test statistics. We denote the two variant test statistics, the Kolmogorov--Smirnov type and the Cram\'{e}r--von Mises type, by $S_1$ and $S_2$, respectively. We implement these methods using the IntRegGOF library in the R package. One obvious drawback of $S_1$ and $S_2$ is that they are sensitive to the number of predictors. One possible way to reduce the effect of the dimension of the predictor is to use a test indexed by certain projections of the predictor; see the test based on $W_p$ in page 1394 of \cite{StuteEtAl06}. We also implement this test and denote it by $S_P$. As $S_P$ is based solely on one projected direction the derived test can handle more predictors but the test need not have high power against all alternatives.

\cite{Fan01} proposed a lack-of-fit test based on Fourier transforms under the assumption of independent and identically distributed Gaussian errors; also see \cite{Christensen10} for a very similar method. The main drawback of this approach is that the method needs a reliable estimator of $\mathrm{var}(\eta)$ to compute the test statistic, and it can be very difficult to obtain such an estimator under model mis-specification. 

We present the power study of the adaptive Neyman test $T_{AN,1}^*$ of~\cite{Fan01} using the known $\mathrm{var}(\eta)$ as a gold standard; see equation (2.1) of the paper. We denote this test statistic by $F$. When using an estimate of $\mathrm{var}(\eta)$, as in equation (2.10) of~\cite{Fan01}, we got very poor results. 

\cite{PS06} proposed an easy-to-implement single global procedure for testing the various assumptions of a linear model. Their test can be viewed as a Neyman smooth test and relies only on the standardized residual vector. We implemented their procedure using the gvlma library in the R package and denote it by $G$. We also implement the generalized likelihood ratio test of~\cite{FJ07}; see equation (4.24) of their paper and also~\cite{FJ05}. The test computes the likelihood ratio statistic, assuming normal errors, obtained from the parametric and nonparametric fits. As the procedure involves fitting a nonparametric model, it requires a delicate choice of smoothing bandwidths. We use the np library in the R package to compute the nonparametric kernel estimator with the optimal bandwidth being chosen by the npregbw function in that package. This procedure is similar in spirit to that used in~\cite{HM93}. To compute the critical value of the test we use the wild bootstrap method.

From Tables 2 and 3 it is clear that our procedure overall has much better finite sample performance than the competing methods. As $a$ increases, the power of our test monotonically increases to 1 in all the simulation settings. It even performs better than $F$, which uses the known $\mathrm{var}(\eta)$, in most cases.  As expected, $S_1$ and $S_2$ behave poorly as the dimension of the predictor increases, whereas $S_P$ does not show any such deterioration in performance. However, as seen from the tables, $S_P$ is slow to capture the departure from $H_0$ as $a$ increases. This is a drawback of using only one projected direction of the predictor. The method of~\cite{FJ07}, $L$, is anti-conservative, drastically violates the level condition, and hence shows higher power in some scenarios. It is also computationally expensive as it involves the choice of smoothing parameters, especially for higher dimensional predictors. 

\subsection{Real data analysis}
%We consider two regression data sets and illustrate the performance of our procedure on them. 
\begin{figure}
\begin{center}
\includegraphics[height=3.0in,width= 2.87in]{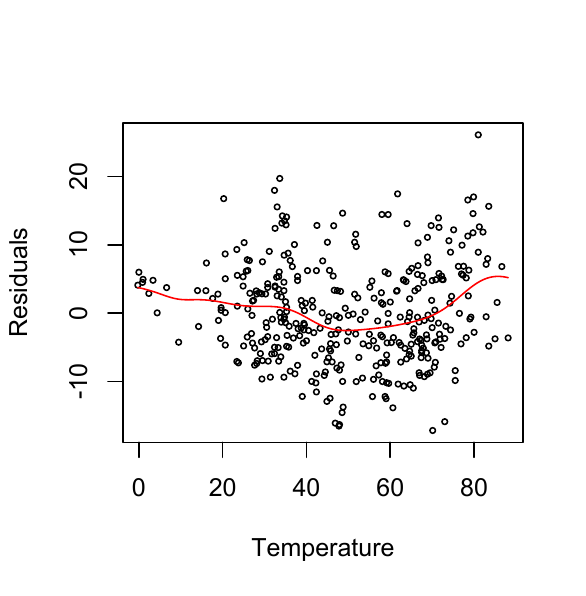}
\includegraphics[height=3.0in,width= 2.87in]{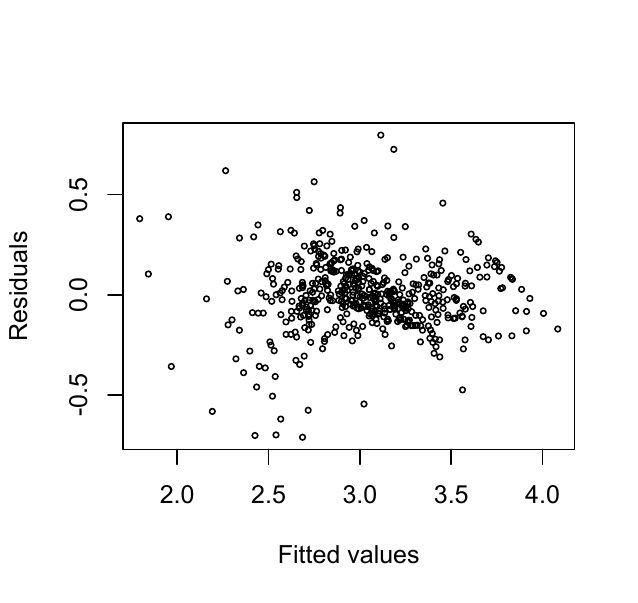}
%\figurebox{15pc}{15pc}{180}[ChicagoPredResidPlot.eps]
%\figurebox{20pc}{25pc}{}[BostonResidPlot.pdf]
\caption{(a) Plot of residuals against temperature for Example 1 and the corresponding  nonparametric regression. \ \ \ \ (b) Plot of residuals against fitted values for Example 2.}
\label{fig:DiagPlots}
\end{center}
\end{figure}
\begin{example} The first data set involves understanding the relation between the atmospheric ozone level and a variety of atmospheric pollutants, e.g.\ nitrogen dioxide, carbon dioxide, sulphur dioxide, etc.,  and weather conditions, including daily temperature and humidity. The data set contains daily measurements for the year 1997 on 9 variables, and is studied in \cite{Xia09}. For a complete background on the data set see the reports of the World Health Organization (2003), Bonn, Switzerland; the data set is available at http://www.ihapss.jhsph.edu/data/data.htm. As illustrated in \cite{Xia09}, the data exhibit a non-linear trend. Figure~\ref{fig:DiagPlots}(a) shows the residuals, obtained from the fit in equation (2) of \cite{Xia09}, against temperature, and clearly illustrates the dependence of the residuals on the predictor. However, neither~\cite{StuteEtAl98} nor~\cite{Fan01} reject the linear model specification at $5$\% significance level, which implies that their methods are not efficient with multiple regressors. Our procedure yields a $p$-value of $0.02$. 
\end{example}
\begin{example} We study the Boston housing data, collected by \cite{HR78} to study the effect of air pollution on real estate prices in the greater Boston area. The data consist of 506 observations on 16 variables, with each observation pertaining to one census tract. We use the version of the data that incorporates the minor corrections found by \cite{GP96}. Figure~\ref{fig:DiagPlots}(b) shows the residual plot for the model fitted by \cite{HR78}, which clearly exhibits heteroscedasticity. Our procedure yields a $p$-value of essentially 0 while the method of \cite{StuteEtAl98} yields a $p$-value of over 0.2.
\end{example}
%{\bf Example 3:} The other example we consider is the {\it savings} data set given in \cite{Faraway05} (see Chapter 3, page 31). This data set consists of some economic measurements collected for 50 countries and has 4 predictors. It is used in the book as an illustration of the various inferential techniques using multiple linear regression. For this data our method, along with that of \cite{StuteEtAl98} and \cite{Fan01}, accepts the null hypothesis \eqref{eq:H_0}. We observe a $p$-value of $0.2$ for our method. Thus, there is a natural agreement among the competing procedures, as can be expected from a data set used to demonstrate multiple linear regression.

\section{Appendix A}\label{AppA}
\subsection{A general theorem for triangular arrays}\label{ConvRes}
Instead of proving the convergence of $T_n$ under the null hypothesis and the consistency of our bootstrap procedure separately we here present a general result involving triangular arrays of random variables from which Theorems~\ref{thm:LimDist-T_n} and~\ref{thm:BootsCons} will easily follow.

We denote by $Z = (X, \eps) \sim P$ on $\R^{d_0} \times \R$. For each $n \ge 1$, we will consider a triangular array of random vectors $Z_{in} = (X_{in}, \eps_{in})$ for $i = 1,\ldots, n, $ independent and identically distributed from a distribution $P_n$ on $\R^{d_0} \times \R$. For $\beta_n \in \R^d$ define 
\begin{eqnarray}\label{eq:TriModel}
Y_{in} = g(X_{in})^\T \beta_n + \eps_{in} \quad (i = 1, \ldots, n). \nonumber
\end{eqnarray} 
%We will denote by $E_n$ the expectation with respect to the probability measure $P_n$. 
We may assume that  the random vectors $Z,  Z_{in}$ for $i=1, \ldots, n$,  and $n = 1, 2, \ldots,$ are all defined on a common probability space. 

%and  we denote by $\prob$ and $\E$ the probability measure and the corresponding expectation operator on that  probability space. 

We compute an estimator $\beta_n^*$ of $\beta_n$ using the method of least squares, i.e.,
\[ \beta_n^*  = \text{argmin}_{\beta \in \mathbb{R}^d} \sum_{i=1}^n \big \{Y_{in} -  g(X_{in})^\T \beta \big\}^2 = A_n^{-1} \left\{ {n}^{-1} \sum_{i=1}^n g(X_{in}) Y_{in} \right\}, \]
where $ A_n = {n}^{-1}\sum_{i=1}^n g(X_{in}) g(X_{in})^\T$ is assumed to be invertible. Write 
\begin{eqnarray*}
\eps_{in}^* = Y_{in} - g(X_{in})^\T \beta_n^* 
\end{eqnarray*} 
for the $i$-th residual at stage $n$. We want to find the limit distribution of the statistic 
\[ T_n^* = \frac{1}{n^2} \sum_{i, j}^n k_{ij} l_{ij}^*  + \frac{1}{n^4} \sum_{i, j, q, r}^n k_{ij} l_{qr}^* -  \frac{2}{n^3} \sum_{i, j, q}^n k_{ij} l_{iq}^*,\]
where $k:\R^{d_0} \times \R^{d_0} \to \R, l:\R \times \R \to \R$ are kernels, $k_{ij}  = k(X_{in}, X_{jn}),$ and $l_{ij}^* = l(\eps_{in}^* ,\eps_{jn}^*)$. We make the following assumptions to study the limiting behavior of $T_n^*$.
\begin{condition} \label{c:C1}
Assume the following conditions on the measures $P_n$:

(a)  $X_{in} $ and $\eps_{in}$ are independent. In other words,  $P_n = P_{n,X} \times P_{n,\eps}$, for all $n$, where $P_{n,X}$ is a  measure on $\mathbb{R}^{d_0}$ and  $P_{n, \eps}$ is a measure on $\mathbb{R}$; 

(b)  $\E (\eps_{1n}) = 0 \quad (n = 1,2, \ldots)$;
	%, and $\sigma^2_n = E[\eps_{in}^2] \rightarrow \sigma^2 = E[\eps^2]$
	
(c) there exists a distribution $P = P_X \times P_\eps$ on $\mathbb{R}^{d_0} \times \mathbb{R}$ such that  $P_n \to P$, in distribution;

(d)  $\{X_{1n}, g(X_{1n}) \} \to \{X, g(X)\}$ in distribution,  where $X \sim P_X$.
\end{condition}
	
%that is C1.1-4	
\begin{condition} \label{c:C2}
 The following families of random variables are uniformly integrable for any $1 \le p, q, r,s \le 4$,

		(a) $\{ | g(X_{pn}) |^2_\infty  : n \ge 1 \}$,
		
		(b)  $ \{ | \eps_{pn} | ^ 2 : n \ge 1\} $,
		
		(c) $ \big\{   k^2(X_{pn}, X_{qn})(1 + |g(X_{rn})|^2_\infty )(1+ |g(X_{sn})|^2_\infty )  : n \ge 1 \big\}$,

		(d) $\{  f^2(\eps_{pn}, \eps_{qn}): n \ge 1\}  \quad  (f = l, l_x, l_y, l_{xx}, l_{yy}, l_{xy})$. 
\end{condition}

\begin{theorem} \label{main_thm}
Suppose that Conditions~\ref{c:I}, \ref{c:K}, \ref{c:C1} and \ref{c:C2} hold. Then $nT_n^* \to \chi \equiv \chi(P_X \times P_\eps)$, in distribution,  where $\chi$ is described in \textsection~\ref{DescLD} with $\eta$ replaced by $\eps$.
\end{theorem}

\subsection{Proofs of theorems}
Theorem~\ref{thm:LimDist-T_n} is an easy consequence of Theorem~\ref{main_thm}, by taking $P_n \equiv P$ for all $n$.  Under $H_0$, $P$ is in the product form  $P_X \times P_\eta$ which implies Condition~\ref{c:C1}(a). Condition~\ref{c:C1}(b)--(d)   are also trivially satisfied. Moreover,  Condition~\ref{c:C2} is immediate from Condition~\ref{c:M}.

%The proof of Theorem~\ref{main_thm} is rather long and  can be found in the supplementary material \cite{supplement}  along with the deduction of  Theorem~\ref{thm:BootsCons}  from Theorem~\ref{main_thm}.  
Next we give a sketch of the proof of Theorem~\ref{thm:LimDistAlt-T_n}. 

\begin{proof}[of Theorem~\ref{thm:LimDistAlt-T_n}]
Let  $\eps_{i} = m(X_i) - g(X_i) ^\T \tilde \beta_0 + \eta_i$. The least squares estimator $\hat \beta_n$  admits the following expansion around $\tilde \beta_0$:
\begin{equation}\label{beta_n_expansionCLT}
n^{1/2}( \hat \beta_n - \tilde \beta_0) =   \{I +o_p(1)\} n^{-1/2} \sum_{i=1}^n A^{-1}  g(X_i) \eps_i.
\end{equation}
 The normal equation for the regression model yields $\E[ g(X) \{ m(X) - g(X) ^\T \tilde \beta_{0}\}]=0.$
 Also, $\E \{ g(X) \eta \} = \E \{ g(X)\E (\eta \mid X) \} =  0$.  Hence, we have $\E \{ g(X)\eps \}=0$. So, by the central limit theorem  $n^{1/2}( \hat \beta_n - \tilde \beta_0) $ converges in distribution  to a Gaussian random  vector with mean $0$ and covariance $A^{-1} \E \{  g(X)g(X)^\T \eps^2 \} A^{-1}$.
We  expand $l_{ij} =  l(e_{i} ,e_{j}) $  around $l(\eps_i, \eps_j)$ using Taylor's theorem as 
\begin{eqnarray*}
l_{ij}  =  l(\eps_{i}, \eps_{j}) + \Big \{ (e_i - \eps_{i} ) l_x(\gamma_{ijn}, \tau_{ijn}) + (e_{j} - \eps_{j} ) l_y(\gamma_{ijn}, \tau_{ijn}) \Big\}
\end{eqnarray*}
for some point $(\gamma_{ijn},\tau_{ijn})$ on  the line joining $(e_{i}, e_j)$ and $(\eps_{i},\eps_{j})$. 
We can decompose $T_n$ as 
$$ T_n = T_{n}^{(0)} + (\hat \beta_n - \tilde \beta_{0})^\T T_{n}^{(1)} + R_n,$$ where 
\begin{eqnarray*}
T_{n}^{(p)} & = & \frac{1}{n^2} \sum_{i, j}^n k_{ij} l_{ij}^{(p)}  + \frac{1}{n^4} \sum_{i, j, q, r}^n k_{ij} l_{qr}^{(p)} - 2 \frac{1}{n^3} \sum_{i, j, q}^n k_{ij} l_{iq}^{(p)}   \quad (p=0,1), \\
l_{ij}^{(0)} & = &   l(\eps_{i}, \eps_{j}), \qquad   l_{ij}^{(1)}  =  - \Big\{ l_x(\eps_i, \eps_j) g(X_{i}) +  l_y(\eps_i, \eps_j) g(X_{j}) \Big\}.
\end{eqnarray*}
 It can be shown that $n^{1/2} R_n \to 0,$ in probability.

Thus it remains to find the limiting distribution of $T_{n}^{(0)} + (\hat \beta_n - \tilde \beta_{0})^\T T_{n}^{(1)}$. Under each of $H_1, H_2$ and $H_3$,   $X$ and $\eps$ are not independent  and hence $\theta(X, \eps)>0$ where $\theta(X, \eps)$ is the Hilbert--Schmidt independence criterion of the joint distribution $(X, \eps)$. 
Letting $W_i = (X_i, \eps_i)$,  $T_{n}^{(p)}$ can naturally be written as a $V$-statistic
\begin{equation*}
 T_{n}^{(p)} = \frac{1}{n^4} \sum_{1 \le q, r, s, t \le n} h^{(p)}(W_{q}, W_{r},  W_{s}, W_{t}) \quad (p=0,1), 
 \end{equation*}
for some symmetric kernel 
\begin{eqnarray*}
h^{(p)}(W_{q}, W_{r},  W_{s}, W_{t}) = \frac{1}{4!} \sum_{(i,j, u, v)}^{(q,r,s,t)}  k_{ij} l_{ij}^{(p)} + k_{ij} l_{uv}^{(p)}  - 2 k_{ij}l_{iu}^{(p)}, \label{eq:h^p}
\end{eqnarray*}
where the sum is taken over all $4!$ permutations of $(q,r,s,t)$. By the definition of $\theta$, 
  $\E \{h^{(0)}(W_{1}, W_{2},  W_{3}, W_{4})\} = \theta(X, \eps)$.
Thus  from standard theory of V-statistics, we obtain
\begin{equation}\label{eq:Tn0_alt}
 {n}^{1/2} \{T_n^{(0)}  - \theta(X, \eps)\}= n^{-1/2} \sum_{i=1}^n h^{(0)}_1(W_i) + o_p(1), 
 \end{equation}
where $h^{(0)}_1(w) =  \E \{h^{(0)}(w , W_{2},  W_{3}, W_{4})\} -  \theta(X, \eps)$ such that  $\E \{h^{(0)}_1(W_1)\} = 0$. 
On the other hand,  by the weak law of large numbers  for V-statistics,
\[ T_n^{(1)}  \to \gamma = \E \{ h^{(1)}(W_{1}, W_{2},  W_{3}, W_{4})\},  \]
in probability. 
From \eqref{beta_n_expansionCLT} and \eqref{eq:Tn0_alt}, 
\begin{eqnarray*}
	{n}^{1/2} \{T_n - \theta(X,\eps)\}
	& = &  n^{-1/2} \sum_{i=1}^n \left \{ h_1^{(0)}(W_i) + \gamma^\T A^{-1} g(X_i) \eps_i \right\} + o_p(1), 
\end{eqnarray*} 
which by the central limit theorem has an asymptotic  normal distribution with mean $0$ and variance $\mathrm{var} \{h^{(0)}_1(W_1)+ \gamma^\T A^{-1} g(X_1)\eps_1\}.$
\end{proof}

\section{Appendix B}\label{AppB}
This section includes the proofs of Theorems~\ref{thm:BootsCons} and~\ref{main_thm} along with the details of the proof of Theorem~\ref{thm:LimDistAlt-T_n}.

\subsection{Proof of Theorem 3}
We will apply Theorem 4 to derive the desired result by checking  that Conditions~6 and 7  hold for each $\omega \in \Omega$, outside a set of measure zero. We will apply Theorem~4 conditional on  $\mathfrak{Z}$,  and thus the probability and expectation operators in Theorem~4 are now $\mathrm{pr}_\omega$ and $\E_\omega$, respectively.  We will apply the theorem with $\eps_{in} = \eta_{in}^*, $ $X_{in} = X_{in}^*$ $ (i = 1, \ldots, n)$, and with random measures $P_n = P_{n ,X} \times P_{n, e^o}$ where,
\[ P_{n, X} = n^{-1}\sum_{i=1}^n \delta_{X_i}, \quad    P_{n, e^o} =n^{-1} \sum_{i=1}^n \delta_{e^o_i}.\] Define
\begin{equation}\label{def:eps_i}
\eps_i = m(X_i) - g(X_i)^\T \tilde \beta_{0} + \eta_i   \quad (i=1,\ldots,n).
\end{equation}
 Then $\eps_1, \ldots, \eps_n$ are independent and identically distributed. 
%~with mean $ \E [\eps_1] = \E[m(X)] - \E[g(X)]^\T \tilde \beta_{0}$. 
Let  $P_{\eps^o}$ be the distribution of  $\eps_i^o = \eps_i - \E (\eps_i)$.

Let us start by verifying  Condition~6. By definition, $P_n = P_{n,X} \times P_{n,e^o}$ is a product measure.  We take $P = P_{X} \times P_{\eps^o}$, where $P_X$ and  $P_{\eps^o}$  are the  distributions of $X_i$ and $\eps^o_i$ respectively.  By Lemma~\ref{lem:Cons_epsDist}(ii) below, almost surely,  $P_{n,e^o} \rightarrow P_{\eps^o}$, in distribution. An application of the Glivenko-Cantelli theorem yields that almost surely, $P_{n,X} \rightarrow P_X$, in distribution. Similarly, almost surely, $(X_{1n}^*, g(X_{1n}^*) )  \rightarrow  (X, g(X))$, in distribution.  Also, $\E (\eps_{1n}) = \E_\omega  (\eta_{1n}^*)= \prob_n (e - \bar{e}) = 0$.

We will  now  show that Condition~7 holds. First,   by Lemma~\ref{lem:Cons_epsDist}(iii) below, 
$$\E_\omega ({|\eta_{1n}^*|^{2+\delta}}) = \prob_n (|e^o|^{2 +\delta}) = O_\omega(1).$$  This shows Condition~7(b). Condition~7(a) holds, by assumption Condition~5(a)  and since by the strong law of large numbers,  almost surely, $$\E_\omega \{{|g(X_{1n}^*)|_\infty^{2 +\delta}} \} = \mathbb{P}_n \{|g(X)|_\infty^{2 +\delta}\} \rightarrow \E \{|g(X)|_\infty^{2 +\delta}\} < \infty.$$ To verify Condition~7(c), notice that the quantity of interest is a V-statistic. The strong law of large numbers for U-statistics along with  Condition~5(c)  implies that Condition~7(c) holds.

It remains to check Condition~7(d). Throughout the rest of proof, we will use  the notation `$a_n \lesssim b_n$'  for two positive sequences of real numbers $a_n$ and $b_n$  to mean that $a_n \le C b_n$, for all~$n$ for some  constant $C$. Consider $f = l_{xx}, l_{xy}$ or $l_{yy}$. Then, for $q \ne r$,  \[ \E_\omega \{ |f(\eta_{qn}^*,\eta_{rn}^*)|^{2 + \delta}\} = n^{-2} \sum_{i, j=1}^n |f(e_i^o, e_j^o)|^{2 + \delta}, \]
which  can be bounded by
\begin{eqnarray}
& &  2^{2 + \delta} n^{-2} \sum_{i, j=1}^n \left \{   \left | f(e_i^o, e_j^o) - f(\eps_i^o, \eps_j^o) \right |^{2 + \delta} +  \left| f(\eps_i^o, \eps_j^o) \right|^{2 + \delta} \right \}  \nonumber \\  
&\lesssim & \mathbb{P}_n (  |e^o - \eps^o|^{2 + \delta} )  + n^{-2}\sum_{i, j=1}^n  \left| f(\eps_i^o, \eps_j^o) \right|^{2 + \delta}  =   O_\omega(1). \label{eq:Bound_f}
\end{eqnarray}
In the  inequality above, we have used the Lipschitz continuity of $f$.  By  the strong law of large numbers  for V-statistics $n^{-2}\sum_{i, j=1}^n  | f(\eps_i^o, \eps_j^o) |^{2 + \delta} \to \E\left \{ | f(\eps_1^o, \eps_2^o) |^{2 + \delta} \right\}$, almost surely, 
 which holds under the moment condition  $\E \{ |f(\eps_q^o, \eps_r^o) |^{2+\delta} \} < \infty$  for  $f =l_{xx}, l_{xy}$ or $l_{yy}$ from Condition~5(d). This fact along with Lemma \ref{lem:Cons_epsDist}(i) below justifies the equality in  \eqref{eq:Bound_f}.
 
A similar analysis can be done for the case $q = r$. Indeed, $\E_\omega \{ |f(\eta_{qn}^*,\eta_{qn}^*)|^{2 + \delta}\} = n^{-1}\sum_{i=1}^n  |f(e_i^o, e_i^o)|^{2 + \delta}$ is  bounded by
\begin{eqnarray*}
& & 2^{2 + \delta}n^{-1} \sum_{i=1}^n \left \{\ \left| f(e_i^o, e_i^o) - f(\eps_i^o, \eps_i^o) \right|^{2 + \delta} +  \left| f(\eps_i^o, \eps_i^o) \right|^{2 + \delta} \right \}  \\
&\lesssim & \mathbb{P}_n (  |e^o - \eps^o|^{2 + \delta} )  + n^{-1}\sum_{i=1}^n  \left| f(\eps_i^o, \eps_i^o) \right|^{2 + \delta}  =   O_\omega(1).
\end{eqnarray*}

Now consider  $f = l_{x}$ or $l_y$. Let $a_i =|e^o_i - \eps^o_i|$ for $i=1, \ldots, n$.  Consider  the following upper bound for $|f(e^o_i , e^o_j)|$ which uses a  one term Taylor expansion for $f$ and the Lipschitz continuity of  the partial  derivatives $f_{x}$ and $f_{y}$:
\begin{equation}\label{ineq:l_x_lip_bdd}
 |f(e^o_i , e^o_j))| \le  |f(\eps_i^o, \eps_j^o)| + a_i |f_{x}(\eps_i^o, \eps_j^o)| + a_j  |f_{y}(\eps_i^o, \eps_j^o) |  + 2L (a_i+ a_j).
\end{equation}
Consequently, if  $q \ne r$,  $  \E_\omega \{ |f(\eta_{qn}^*,\eta_{rn}^*)|^{2 + \delta}\}$ is bounded from above,  up to a constant, by 
\begin{eqnarray*}
n^{-2} \sum_{i, j}  | f(\eps_i^o, \eps_j^o)|^{2+\delta} + n^{-2} \sum_{i, j} \left \{ |a_i f_{x}(\eps_i^o, \eps_j^o)|^{2+\delta} +  |a_j  f_{y}(\eps_i^o, \eps_j^o) |^{2+\delta} \right\} +  \mathbb{P}_n (|a|^{2 +\delta}).
\end{eqnarray*}
The first and the  third term are  $O_\omega(1)$ by  Condition~5(d) and Lemma~\ref{lem:Cons_epsDist}(i) below.  Further, 
\begin{align*}
 a_i  \le d |\hat \beta_n  - \tilde \beta_0|_\infty  |g(X_i)|_\infty  + |\bar e - \E(\eps) |  = O_\omega(1) \{1+ |g(X_i)|_\infty\}.
 \end{align*}
 Therefore, 
\[ n^{-2} \sum_{i, j} |a_i f_{x}(\eps_i^o, \eps_j^o)|^{2+\delta} \le   O_\omega(1)  n^{-2} \sum_{i, j} \left|  \{1+ |g(X_i)|_\infty\} f_{x}(\eps_i^o, \eps_j^o)\right|^{2+\delta}, \]
which is again $O_\omega(1)$ by  the strong law of large numbers for  V-statistics which holds under Condition~5(d).  Similarly, 
$ n^{-2} \sum_{i, j} |a_j f_{y}(\eps_i^o, \eps_j^o)|^{2+\delta}  = O_\omega(1).$
Putting these together,  we obtain  that 
$$ \E_\omega \{ |f(\eta_{qn}^*,\eta_{rn}^*)|^{2 + \delta} \}= O_\omega(1) \quad  (q \ne r).$$ A similar analysis  shows that $\E_\omega \{|f(\eta_{qn}^*,\eta_{qn}^*)|^{2 + \delta}\} = O_\omega(1)$.

%On the other hand, if  $q  =  r$, we can bound $  \E_\omega[|f(\eta_{qn}^*,\eta_{rn}^*)|^{2 + \delta}$ by, up to a constant,  
%\begin{eqnarray*}
%\frac{1}{n^2} \sum_{i}  |l_x(\eps_i^o, \eps_i^o)|^{2+\delta} + \frac{1}{n} \sum_{i} (|a_i l_{xx}(\eps_i^o, \eps_i^o)|^{2+\delta} +  |a_i  l_{xy}(\eps_i^o, \eps_i^o) |^{2+\delta} ) +  \mathbb{P}_n [|e - \eps|^{4 +\delta}],
%\end{eqnarray*}
%which is $O_\omega(1)$. 

For $f = l$, we can closely imitate  the above argument for $f = l_x$ or $l_y$    to deduce  that $\E_\omega \{ |f(\eta_{qn}^*,\eta_{rn}^*)|^{2 + \delta}\} = O_\omega(1)$ for any $1 \le q, r \le 2$. We just need to replace  \eqref{ineq:l_x_lip_bdd} with  the following inequality which  follows from the two-term Taylor expansion of the function $l$:
\begin{eqnarray*}
 |l(e^o_i , e^o_j))|  \le && |l(\eps_i^o, \eps_j^o)| + a_i |l_{x}(\eps_i^o, \eps_j^o)| + a_j  |l_{y}(\eps_i^o, \eps_j^o) |  + \tfrac{1}{2}  a_i^2 |l_{xx}(\eps_i^o, \eps_j^o)|  \\
&&+\; \tfrac{1}{2}  a_j^2 |l_{yy}(\eps_i^o, \eps_j^o)|  + a_i  a_j |l_{xy}(\eps_i^o, \eps_j^o)| +  4L (a_i^2+ a_j^2).
\end{eqnarray*}
We omit the routine details. Thus Condition~7(d) of  Theorem~4 holds.  This concludes the proof of  Theorem~3. \qed

\subsection{Proof of Theorem~2}
Let  $\eps_{i}$ be as defined in \eqref{def:eps_i}. The least squares estimator $\hat \beta_n$  admits the following expansion around $\tilde \beta_0$:
\begin{eqnarray}\label{beta_n_expansionCLT}
n^{1/2}( \hat \beta_n - \tilde \beta_0) & = & n^{1/2} \left \{ A_n^{-1} n^{-1} \sum_{i=1}^n g(X_{i}) Y_{i} - \tilde \beta_0 \right\}  \notag \\
& = & {n}^{1/2} \left \{  A_n^{-1} n^{-1} \sum_{i=1}^n g(X_{i}) (m(X_{i}) - g(X_{i}) ^\T \tilde \beta_{0} + \eta_{i})  \right\}, \notag\\
& = &   \{I +o_p(1)\} n^{-1/2} \sum_{i=1}^n A^{-1}  g(X_i) \eps_i,
\end{eqnarray}
where in the last step we have used the fact that $A_n \to A$, almost surely, which holds as $\E\{|g(X)|_\infty^2\} < \infty$. 
%By \eqref{beta_n_expansionCLT}  $\hat \beta_n$ admits the following expansion around $\tilde \beta_0$:
%\[ \sqrt{n}( \hat \beta_n - \tilde \beta_0)  =    (I+o_\prob(1)) n^{-1/2} \sum_{i=1}^n A^{-1}  g(X_i) \eps_i.\]
 The normal equation for the regression model is $$\E[ g(X) \{ m(X) - g(X) ^\T \tilde \beta_{0}\}]=0.$$
 Also, $\E \{g(X) \eta\} = \E \{ g(X)\E (\eta\mid X) \} =  0$.  Hence, we have $\E \{ g(X)\eps\}=0$. Moreover,  Condition~4(b), the covariance matrix $A^{-1} \E \{  g(X)g(X)^\T \eps^2\} A^{-1}$ exists. So, by the central limit theorem, ${n}^{1/2}( \hat \beta_n - \tilde \beta_0) $ converges in distribution  to a Gaussian random  vector with mean $0$ and covariance $A^{-1} \E \{ g(X)g(X)^\T \eps^2\} A^{-1}$.

We  expand $l_{ij}  =  l(e_{i} ,e_{j}) $  around $l(\eps_i, \eps_j)$ using Taylor's theorem as 
\begin{eqnarray*}
l_{ij}  =  l(\eps_{i}, \eps_{j}) + \Big \{(e_i - \eps_{i} ) l_x(\gamma_{ijn}, \tau_{ijn}) + (e_{j} - \eps_{j} ) l_y(\gamma_{ijn}, \tau_{ijn}) \Big\}
\end{eqnarray*}
where  $(\gamma_{ijn},\tau_{ijn})$ is  some point on  the line joining $(e_{i}, e_j)$ and $(\eps_{i},\eps_{j})$. 
Using
\begin{eqnarray}\label{eq:e_eps_diff} 
	e_{i} - \eps_{i} & = & -g(X_{i})^\T (\hat \beta_n - \tilde \beta_{0}), 
\end{eqnarray}
decompose $T_n$ in the following way:
$$ T_n = T_{n}^{(0)} + (\hat \beta_n - \tilde \beta_{0})^\T T_{n}^{(1)} + R_n,$$ where
\begin{eqnarray*}
T_{n}^{(p)} & = & \frac{1}{n^2} \sum_{i, j}^n k_{ij} l_{ij}^{(p)}  + \frac{1}{n^4} \sum_{i, j, q, r}^n k_{ij} l_{qr}^{(p)} -  \frac{2}{n^3} \sum_{i, j, q}^n k_{ij} l_{iq}^{(p)}  \quad (p = 0,1),
\end{eqnarray*}
 and 
\[ l_{ij}^{(0)} =  l(\eps_{i}, \eps_{j}), \qquad   l_{ij}^{(1)}  =  - \Big \{ l_x(\eps_i, \eps_j) g(X_{i}) +  l_y(\eps_i, \eps_j) g(X_{j}) \Big\}. \]
 We will first show  the negligibility of the reminder term $R_n$. More precisely, we claim that  ${n}^{1/2} R_n \to 0$, in probability. To prove the claim we need the following elementary lemma which we state without proof.
\begin{lemma}\label{lem:quadratic}
Let $f: \R^2 \to \R$ be a continuously differentiable function with its partial derivatives $f_x, f_y$ being Lipschitz continuous with Lipschitz constant $L$ with respect to $\ell_\infty$ norm. Then for any $u, v \in \R^2$,  
\[ |f(v) - f(u)| \le 2| \nabla f(u)|_\infty   |u-v|_\infty + 2L |u-v|_\infty^2. \]
\end{lemma}
An application of the above lemma together with \eqref{eq:e_eps_diff}  gives
\[  \begin{split}
&| l_x(\gamma_{ijn}, \tau_{ijn})    - l_x(\eps_i, \eps_j)  |_\infty  \lesssim  |\nabla l_{x}(\eps_i, \eps_j) |_\infty  | \hat \beta_n   - \tilde \beta_0|_\infty  \big \{ |g(X_{i})| _\infty + |g(X_j)|_\infty  \big\}   \\
      &  \hspace{2.3in} + | \hat \beta_n   - \tilde \beta_0|_\infty^2  \big\{ |g(X_{i})| _\infty + |g(X_j)|_\infty  \big\}^2.
\end{split} \]
Similarly, we can bound $| l_y(\gamma_{ijn}, \tau_{ijn})   - l_y(\eps_i, \eps_j)  |_\infty.$
Finally,  we can bound ${n}^{1/2}|R_n|$, up to a constant,  by 
\begin{equation} \label{eq:remainder_bdd_h13}
   {n}^{1/2}  | \hat \beta_n   - \tilde \beta_0|_\infty^2 T_n^{(2)} +  {n}^{1/2}  | \hat \beta_n   - \tilde \beta_0|_\infty^3 T_n^{(3)}, 
  \end{equation}
where, $T_n^{(2)}$ and $T_n^{(3)}$ are defined as follows:
\begin{eqnarray*}
T_{n}^{(p)} & = & n^{-4} \sum_{i, j, q , r}^n |k_{ij}|  \big( l_{ij}^{(p)}  + l_{qr}^{(p)} +  l_{iq}^{(p)} \big)  \quad  (p = 2,3),
\end{eqnarray*}
with
\begin{eqnarray*}
l_{ij}^{(2)}  =   |\mathrm{Hess}(l)(\eps_i, \eps_j)|_\infty  \big \{ |g(X_{i})|^2_\infty + |g(X_{j})|^2_\infty\big\}, \quad 
l_{ij}^{(3)} =   |g(X_{i})|^3_\infty  + |g(X_{j})|^3_\infty.
\end{eqnarray*}
Clearly, $T_n^{(2)}$ and $T_n^{(3)}$ are V-statistics whose kernels are integrable by Condition~4(c)(iii)--(iv). Consequently, the weak law of large numbers for V-statistics holds for $T_n^{(2)}$ and $T_n^{(3)}$. 
Now since ${n}^{1/2} | \hat \beta_n  - \tilde \beta_0|_\infty  = O_p(1)$, it follows that \eqref{eq:remainder_bdd_h13}  is $o_p(1)$ and the claim is established. 

Thus it remains to find the limiting distribution of $T_{n}^{(0)} + (\hat \beta_n - \tilde \beta_{0})^\T T_{n}^{(1)}$. To do that  first we will show that  $X$ and $\eps$ are not independent under each of $H_1, H_2$ and $H_3$ and hence $\theta(X, \eps)>0$ where $\theta(X, \eps)$ is the Hilbert--Schmidt independence criterion of the joint distribution $(X, \eps)$. Under hypothesis $H_1$, $X  \not \Perp \eta $ and  $\eps = \eta$. Hence $X \not \Perp \eps$ under $H_1$.   For the case $H_2$ and $H_3$ we proceed as follows. The conditional mean of $\eps$  given $X$ is
\[ \E (\eps \mid X)  =  m(X) - g(X)^\T \tilde \beta_0 + \E (\eta \mid X) = m(X) - g(X)^\T \tilde \beta_0. \]
Under $H_2$ or $H_3$, $m(X) \ne g(X)^\T \tilde \beta_0$ with positive probability. In the case when $m(X) - g(X)^\T \tilde \beta_0$ is a non-constant function of $X$,   $\E(\eps\mid X)$  depends on $X$, and hence $X$ and $\eps$ are not independent. The case $m(X)  =  g(X)^\T \tilde \beta_0 +c$ for some non-zero constant $c$  does not arise for $H_2$ by the assumption in Theorem~2. On the other hand, under $H_3$, if $m(X)  =  g(X)^\T \tilde \beta_0 +c$, then $\eps  = c+ \eta$. Thus $\eps$ and $X$ are not independent. 

Let $W_i = (X_i, \eps_i)$. Then   $T_{n}^{(p)} (p=0, 1)$ can naturally be written as a $V$-statistic: 
\begin{equation*}
 T_{n}^{(p)} = n^{-4} \sum_{q, r, s, t }^n h^{(p)}(W_{q}, W_{r},  W_{s}, W_{t}),
 \end{equation*}
for some symmetric kernel $h^{(p)}$ given by 
\begin{eqnarray*}
h^{(p)}(W_{q}, W_{r},  W_{s}, W_{t}) = \frac{1}{4!} \sum_{(i,j, u, v)}^{(q,r,s,t)}  k_{ij} l_{ij}^{(p)} + k_{ij} l_{uv}^{(p)}  - 2 k_{ij}l_{iu}^{(p)}, \label{eq:h^p}
\end{eqnarray*}
where the sum is  over all $4!$ permutations of $(q,r,s,t)$. 
Under each of the hypotheses $H_1, H_2$ or $H_3$,  $\E \{|h^{(0)}(W_{q}, W_{r},  W_{s}, W_{t})|^2\} < \infty$ for $1 \le q,r,s,t \le 4$ by Condition~4(c)(i). Also,  $\E \{h^{(0)}(W_{1}, W_{2},  W_{3}, W_{4}) \} = \theta(X, \eps)$ by the definition of $\theta$.
Thus appealing to the standard theory of V-statistics, we obtain
\begin{equation}\label{eq:Tn0_alt}
{n}^{1/2} \{T_n^{(0)}  - \theta(X, \eps)\}= n^{-1/2} \sum_{i=1}^n h^{(0)}_1(W_i) + o_p(1), 
 \end{equation}
where $h^{(0)}_1(w) =  \E \{h^{(0)}(w , W_{2},  W_{3}, W_{4})\} -  \theta(X, \eps)$. Also, $\E \{h^{(0)}_1(W_1)\} = 0$ and 
$\E \{ h^{(0)}_1(W_1) ^2\}  \le \mathrm{var}\{h^{(0)}(W_{1}, W_{2},  W_{3}, W_{4})\} < \infty$.

On the other hand, $\E\{ |h^{(1)}(W_{q}, W_{r},  W_{s}, W_{t})|_\infty\} < \infty$ for $1 \le q,r,s,t \le 4$ by  Condition~4(c)(ii)  under hypothesis $H_j$ for each $j=1, 2, 3$. So by the weak law of large numbers for V-statistics, 
\[ T_n^{(1)}  \to \gamma = \E \{ h^{(1)}(W_{1}, W_{2},  W_{3}, W_{4})\},  \]
in probability. From \eqref{beta_n_expansionCLT} and \eqref{eq:Tn0_alt}, 
\begin{eqnarray*}
	{n}^{1/2}\{T_n - \theta(X,\eps)\} & = & {n}^{1/2}\{T_{n}^{(0)} - \theta(X,\eps)\} + {n}^{1/2} (\hat \beta_n - \tilde \beta_{0})^\T T_{n}^{(1)} + o_p(1) \\
	& = &  n^{-1/2} \sum_{i=1}^n \left \{h_1^{(0)}(W_i) + \gamma^\T A^{-1} g(X_i) \eps_i \right\} + o_p(1), 
\end{eqnarray*} 
which by the central limit theorem  has an asymptotic  normal distribution with mean $0$ and variance 
\begin{equation*}\label{eq:sigma_def}
\mathrm{var}\big \{h^{(0)}_1(W_1)+ \gamma^\T A^{-1} g(X_1)\eps_1\big \}. 
\end{equation*}
This concludes the proof  of Theorem~2.

%\begin{eqnarray*}
%&& \E[|l_x(\eps_{qn}, \eps_{rn})|^{2 + \delta} \le 3^{2 + \delta}\left[\E[|l_x(\eps_1 -c, \eps_2 -c)|^{2 + \delta} + |\ - a_0|^{2 + \delta} |l_{xx}(\tilde a, \tilde b)|^{2 + \delta} + |b - b_0|^{2 + \delta} |l_{xy}(\tilde a, \tilde b)|^{2 + \delta} \right]
%\end{eqnarray*}

\subsection{Proof of Theorem~4}
\subsubsection{Decomposition of $T_n^*$}
Observe that
\begin{eqnarray} \label{eq:diff_resid1}
\eps_{in}^* - \eps_{in} 
%&=& g(X_{in})^\T \beta_n - g(X_{in})^\T \beta_n^* \notag \\
 &=&  - (\beta_n^* - \beta_n)^\T g(X_{in}).
\end{eqnarray}
Using  \eqref{eq:diff_resid1} and by Taylor's expansion 
\begin{equation}\label{eq:l_ij*}
l_{ij}^* =  l_{ij}^{(0)}  +  (\beta_n^* - \beta_n)^\T l_{ij}^{(1)} + \frac1 2  (\beta_n^* - \beta_n)^\T v_{ij}^*    (\beta_n^* - \beta_n)
\end{equation}
where 
\begin{align*}
 l_{ij}^{(0)}  &= l_{ij} =  l(\eps_{in}, \eps_{jn}),   \ \ \  l_{ij}^{(1)} = - \Big \{l_x(\eps_{in}, \eps_{jn}) g(X_{in}) + l_y(\eps_{in}, \eps_{jn}) g(X_{jn})\Big \},  \;\;\;\;\;\;\;\;\;\;\;\;\;\;\;\;\;\;\;\;\;\; \\ 
v_{ij}^* &= \Big \{l_{xx}(\vartheta_{ijn}, \tau_{ijn}) g(X_{in}) g(X_{in})^\T + l_{yy}(\vartheta_{ijn}, \tau_{ijn}) g(X_{in}) g(X_{in})^\T \\
&  \;\;\;\;\;\;\;\;\;\;\;\;\;\;\;\;\;\;+\;   2 l_{xy}(\vartheta_{ijn}, \tau_{ijn}) g(X_{in}) g(X_{jn})^\T \Big\},  
\end{align*}
for some  point $(\vartheta_{ijn},  \tau_{ijn}) $ on the straight line connecting the two points $(\eps_{in}^*, \eps_{jn}^*)$ and $(\eps_{in}, \eps_{jn})$ on $\mathbb{R}^2$. 
In view of \eqref{eq:l_ij*}, we can decompose $T_n^*$ in the following way
\begin{equation}\label{eq:T_n*}
 T_n^* = T^{(0)}_{n} + (\beta_n^* - \beta_n)^\T T_{n}^{(1)} +  \frac{1}{2}(\beta_n^* - \beta_n)^\T T_{n}^{(2)} (\beta_n^* - \beta_n) + R_n,
 \end{equation}
  where
\begin{eqnarray*}
T_{n}^{(p)} =  \frac{1}{n^2} \sum_{i, j}^n k_{ij} l_{ij}^{(p)}  + \frac{1}{n^4} \sum_{i, j, q, r}^n k_{ij} l_{qr}^{(p)} -  \frac{2}{n^3} \sum_{i, j, q}^n k_{ij} l_{iq}^{(p)}  \quad (p= 0,1,2),
\end{eqnarray*}
and
\begin{eqnarray*}
l_{ij}^{(2)} & = &  \Big \{ l_{xx}(\eps_{in}, \eps_{jn})  g(X_{in}) g(X_{in})^\T+ l_{yy}(\eps_{in}, \eps_{jn})  g(X_{jn}) g(X_{jn})^\T \\
& & \qquad \qquad +    2l_{xy}(\eps_{in}, \eps_{jn})  g(X_{in}) g(X_{jn})^\T \Big\},
\end{eqnarray*}
and $R_n$ is the reminder term. Here $l_{ij}^{(0)} \in \R$, $l_{ij}^{(1)} \in \R^d$, and $l_{ij}^{(2)} \in \R^{d \times d}$.

For $p \in \{0, 1, 2\}$, $T_{n}^{(p)}$  can  be expressed as a $V$-statistic, although with triangular arrays, of the form
\begin{equation}\label{eq:Vstat_n}
 T_{n}^{(p)} = n^{-4} \sum_{ i, j, q, r}^n h^{(p)}(Z_{in}, Z_{jn},  Z_{qn}, Z_{rn}),
 \end{equation}
for some symmetric kernel $h^{(p)}$ given by 
\begin{equation}\label{def:hp}
h^{(p)}(Z_{in}, Z_{jn},  Z_{qn}, Z_{rn}) = \frac{1}{4!} \sum_{(t, u, v, w)}^{(i, j, q, r)}  k_{tu} l_{tu}^{(p)} + k_{tu} l_{vw}^{(p)}  - 2 k_{tu}l_{tv}^{(p)}, 
\end{equation}
where the sum is  over all $4!$ permutations of $(i, j, q, r)$. $\vspace{0.1in}$

\subsubsection{Getting rid of the triangular sequence}
Let $ Z_i = (X_i, \eps_i) $  be independent and identically distributed random vectors from $P$. By the Skorohod representation theorem, there exists a sufficiently rich probability space $(\widetilde \Omega, \widetilde P)$, independent random elements $\omega_1, \omega_2, \ldots $ defined on $\widetilde \Omega$   and functions $f_n, f$ with  $\tilde Z_{in}=f_n(\omega_i)$,  $\tilde Z_{i} =f(\omega_i)$ such that  $\tilde Z_{in}  =  Z_{in},$ in distribution,   $\tilde Z_{i} =  Z_i$, in distribution,  and almost surely  under $\widetilde P$,  $\tilde Z_{in} \to    \tilde Z_i,$
as $n \to \infty$.
Since we are only concerned about the distributional limit of $nT_n^*$, henceforth in this proof, we may assume, without loss of generality, that  for each $n$, 
the random vectors $W_{in} = \big(Z_{in} , Z_i \big)$ are independent and  for each $i$, $ Z_{in}  \to Z_i $ almost surely as $n \to \infty$. This argument is similar to that in \cite{LeuchtNeumann09}.

We will start by showing that 
\[A_n  = n^{-1} \sum_{i=1}^n g(X_{in}) g(X_{in})^\T  \to  A  = \E\{g(X_{1}) g(X_{1})^\T\}, \]   
in probability. By assumption Condition~7(a),  for any $1 \le p, q \le d$,  $g_p(X_{1n}) g_q(X_{1n})$ are uniformly integrable. Moreover,  by Condition~6(d), we have $g_p(X_{1n}) g_q(X_{1n}) \to g_p(X_{1})g_q(X_{1})$, in distribution. Hence,    $g_p(X_{1n}) g_q(X_{1n}) \to g_p(X_{1}) g_q(X_{1})$ in $L_1$ and  $\E \{ |g_p(X_{1}) g_q(X_{1})|\} < \infty$. Hence,  $n^{-1} \sum_{ i=1}^n g(X_i) g(X_i)^\T \to A$, in probability,  by the weak law of large numbers.  Finally,  
\[  n^{-1}  \sum_{i=1}^n g(X_{in}) g(X_{in})^\T -  n^{-1}  \sum_{i=1}^n g(X_{i}) g(X_{i})^\T \to 0   \]
%\begin{align*}
%n^{-1} \E \left| \sum_{i=1}^n g(X_{in}) g(X_{in})^\T  - g(X_i) g(X_i)^\T  \right| \le \E| g(X_{1n}) g(X_{1n})^\T  - g(X_1) g(X_1)^\T|,
%\end{align*}
in $L_1$ as $n \to \infty$  since $g_p(X_{1n}) g_q(X_{1n}) \to g_p(X_{1}) g_q(X_{1})$ in $L_1$. This completes the proof that $A_n \to A$ in probability.  As a consequence, $A_n$ is invertible, and hence $\beta_n^*$ is well defined  with high probability as $n \to \infty$.

Now $\beta_n^*$ admits the following expansion
\begin{eqnarray}\label{eq:beta_n-beta}
	n^{1/2} (\beta_n^* - \beta_n) &= & n^{-1/2} A_n ^{-1} \sum_{i=1}^n  g(X_{in}) \Big \{Y_{in} - g(X_{in})^\T \beta_n \Big\} \notag \\
	&= & n^{-1/2} A_n ^{-1} \sum_{i=1}^n  g(X_{in}) \eps_{in}.
\end{eqnarray} Next we claim that  
\begin{equation}\label{zetan}
 n^{1/2} (\beta_n^* - \beta_n)  - \zeta_n \to 0,
 \end{equation}
 in probability,  where $\zeta_n =  n^{-1/2} A ^{-1}   \sum_{i=1}^n  g(X_{i}) \eps_{i}$. We will first show that
\begin{equation}\label{zetan_diff}
n^{-1/2}A ^{-1}   \sum_{i=1}^n  g(X_{in}) \eps_{in}  - \zeta_n \to 0
\end{equation}
in $L_2$.
Clearly, it suffices to show that $n^{-1/2}  \sum_{i=1}^n  ( g_p(X_{in}) \eps_{in}  - g_p(X_{i}) \eps_{i})  \to 0$ in $L_2$ for each $1 \le p \le d$.  Indeed,  the square of its $L_2$-norm is
\begin{align*}
n^{-1}  \E \Big [  \sum_{i, j=1}^n  \big\{ g_p(X_{in}) \eps_{in}   - g_p(X_{i}) \eps_{i} \big\} \big\{ g_p(X_{jn}) \eps_{jn}   - g_p(X_{j}) \eps_{j} \big\}  \Big]  \\
=  \E \Big [  \big \{ g_p(X_{1n}) \eps_{1n}   - g_p(X_{1}) \eps_{1} \big\}^2  \Big],  
\end{align*}
which goes to $0$ as $n \to \infty$. This is because   $g_p(X_{1n}) \eps_{1n}  \to  g_p(X_{1}) \eps_{1} $ in distribution  and $ g^2_p(X_{1n}) \eps_{1n}^2$ is uniformly integrable by Conditions~7(a)--(b) and the independence of $X_{1n}$ and $ \eps_{1n}$.  This proves \eqref{zetan_diff}.
Recall that, from \eqref{eq:beta_n-beta},
\[ n^{1/2} (\beta_n^* - \beta_n) =  (A_n ^{-1} A)    n^{-1/2} A ^{-1}   \sum_{i=1}^n  g(X_{in}) \eps_{in}.\]
Since by the central limit theorem, $\zeta_n$ converges in distribution to a multivariate normal,  \eqref{zetan_diff} implies that $n^{-1/2} A ^{-1}   \sum_{i=1}^n  g(X_{in}) \eps_{in} = O_p(1)$. Consequently, 
\[ n^{1/2} (\beta_n^* - \beta_n)  -  n^{-1/2} A ^{-1}   \sum_{i=1}^n  g(X_{in}) \eps_{in} \to 0,  \]
in probability. 
Now \eqref{zetan} follows from \eqref{zetan_diff}. Let $V_n^{(p)}$, for $p  = 0, 1, 2$,  be defined analogously as $T_n^{(p)}$ in \eqref{eq:Vstat_n} but with $Z_{in} = (X_{in}, \eps_{in})$ replaced by $Z_i=(X_i, \eps_i)$. Thus $V_n^{(p)}$ is a proper V-statistic. Our next goal is to show that 
\begin{equation}\label{eq:neg_TtoV}
n^{1 - p/2}(T_n^{(p)}  - V_n^{(p)}) \to 0 \quad (p  = 0, 1, 2), 
\end{equation} 
in $L_2$. 
To show that observe that  \begin{equation*}
\E \left[n^{2 - p} \mathrm{tr} \big\{ (T_n^{(p)}  - V_n^{(p)})( T_n^{(p)}  - V_n^{(p)})^\T \big\} \right] = n^{ -(6+p)} \sum_{\bii, \j} \E [ \mathrm{tr} \big\{   \bar h^{(p)}(\bii) \bar h^{(p)}(\j)^\T \big\}],
\end{equation*}
where $\bii = (i_1, i_2, i_3, i_4)$ and $\j = (j_1, j_2, j_3, j_4)$ are multi-indices in $\{1, \ldots, n\}^4$, and
\[\bar h^{(p)}(\bii) = h^{(p)}(Z_{i_1n} , \ldots, Z_{i_4n})  - h^{(p)}(Z_{i_1} , \ldots, Z_{i_4}).\] 

Let us first show that $|h^{(p)}(Z_{i_1n} , \ldots, Z_{i_4n})|_\infty^{2}$ is uniformly integrable. It is enough to show that each of the terms like $|k_{r s} l^{(p)}_{t u}|_\infty^{2}$, where $ r,s, t, u \in \{1, 2, 3, 4\}$  may not  be necessarily distinct,  is uniformly integrable. Using the independence of $X_{in}$ and $\eps_{in}$, we see that this follows directly from Conditions~7(c)-(d).  Condition~6(d)  together with the continuous mapping theorem implies that, 
\begin{equation*}
h^{(p)}(Z_{i_1n} , \ldots, Z_{i_4n}) \to  h^{(p)}(Z_{i_1} , \ldots, Z_{i_4}),
\end{equation*}
in distribution. Thus the above convergence also holds in $L_2$ and we have that  
 \[ \E \{ |h^{(p)}(Z_{i_1} , \ldots, Z_{i_4})|_\infty^{2}\} < \infty.\]
 
Consequently, $\E \{ |\bar h^{(p)}(\bii)|_\infty^2\}$ is uniformly bounded for all $\bii$ and $n$. An application of the Cauchy-Schwarz inequality yields  
$$\E \{ |\bar h^{(p)}(\bii)|_\infty | \bar h^{(p)}(\j) |_\infty \} \le [\E \{ |\bar h^{(p)}(\bii)|^2_\infty \}]^{1/2} [\E  \{ |\bar h^{(p)}(\j)|_\infty^2\}]^{1/2},$$
implying that $\E \{ |\bar h^{(p)}(\bii) |_\infty |\bar h^{(p)}(\j) |_\infty \}$ is uniformly bounded. 
 The number of multi-indices $\bii$ and $\j$  for which $|\bii \cup \j|  = k $ is bounded above by $n^k$, for each $1 \le k \le 8$. 
 The kernel $h^{(0)}$ is degenerate of order $1$, hence  $\E \{ \bar h^{(0)}(\bii) \bar h^{(0)}(\j)\} = 0$ when $|\bii \cup \j|  =  7 \text { or } 8$.
It will be shown  in Lemma~\ref{l:kernel_h1_mean_0} below that $\E \{h^{(1)}(Z_{1n}, \ldots, Z_{4n})\} = 0$, hence if  $|\bii \cup \j|  =  8$, then $\E \{\bar h^{(1)}(\bii) \bar h^{(1)}(\j)^\T\} = \E \{ \bar h^{(1)}(\bii)\} \E\{ \bar h^{(1)}(\j)^\T\} =  0 $. Putting the above observations together, it remains to prove that 
\[   \E [ \mathrm{tr} \{ \bar h^{(p)}(\bii) \bar h^{(p)}(\j)^\T\} ] \to 0, \]
 for any $\bii, \j$  such that  $ |\bii \cup \j| = 6+p$ and  $p \in \{ 0,1,2 \}$. But this immediately follows from the fact
$\bar h^{(p)}(\bii) \to 0$ in $L_2$
which has  already been shown. Hence \eqref{eq:neg_TtoV} is proved. Finally, we  claim that 
 \begin{equation} \label{eq:triangular_removed}
   \begin{split}
n  \Big \{ T^{(0)}_{n} + (\beta_n^* - \beta_n)^\T T_{n}^{(1)} +  \frac{1}{2}(\beta_n^* - \beta_n)^\T T_{n}^{(2)} (\beta_n^* - \beta_n)  - V_n^{(0)} \\
 -  n^{-1/2}  \zeta_n^\T V_{n}^{(1)}  -   \frac{1}{2} n^{-1} \zeta_n^\T V_{n}^{(2)} \zeta_n \Big\} \to 0,
 \end{split}
 \end{equation} 
in probability, which now easily follows from \eqref{zetan_diff}  and \eqref{eq:neg_TtoV}. $\vspace{0.1in}$

\subsubsection{Negligibility of the reminder term $R_n$} In this subsection, we will show that  the reminder term  can be ignored for future analysis. More precisely, we will prove that 
\begin{equation}\label{eq:Rnto0}
nR_n \to 0,
\end{equation}
in probability.  Let us define 
$$Q_n = \frac{1}{n^2} \sum_{i, j}^n k_{ij} (v_{ij}^* - l_{ij}^{(2)})  + \frac{1}{n^4} \sum_{i, j, q, r}^n k_{ij} (v_{qr}^* - l_{qr}^{(2)}) - \frac{2}{n^3} \sum_{i, j, q}^n k_{ij} (v_{iq}^* - l_{iq}^{(2)}), $$ 
so  that $R_n = (1/2)(\beta_n^* - \beta_n)^\T Q_{n} (\beta_n^* - \beta_n)$. Since  $n^{1/2}(\beta_n^*  - \beta_n) = O_p(1)$ by \eqref{zetan},  
 it is enough to show that  for each $1 \le s, t \le d$, 
 \[ (Q_n)_{st} \to  0,\]
 in probability. 
Note that  $(Q_n)_{st}$ is a sum of three terms and each of these terms  can be shown to converge to $0$ in probability. We will only spell out the details for the first term leaving the other two terms  for the reader. Thus we need to show that  
 \begin{equation} \label{eq:remaninder_first_term}
n^{-2} \sum_{i, j}^n k_{ij} (v_{ij}^* - l_{ij}^{(2)})_{st}  \to 0,
  \end{equation} 
in probability.  The term $(v_{ij}^* - l_{ij}^{(2)})_{st}$ can be further broken down into three terms; the first one being $\big\{l_{xx}(\vartheta_{ijn}, \tau_{ijn}) - l_{xx}(\eps_{in}, \eps_{jn})\big\}g_s(X_{in})g_t(X_{in})$. The other two terms involve $l_{yy}$ and $l_{xy}$. Using the Lipschitz continuity of $l_{xx}, l_{yy}$ and $l_{xy}$ we obtain the following bound: 
\begin{eqnarray*}
|(v_{ij}^* - l_{ij}^{(2)})_{st}| &\lesssim & L|(\eps_{ij}^*,\eps_{ij}^*) -  (\eps_{in},\eps_{jn}) |_\infty \big\{ |g(X_{in})|_\infty + |g(X_{jn})|_\infty \big\}^2  \\ 
& \le & d L |\beta_n^* - \beta_n|_\infty \big\{ |g(X_{in})|_\infty + |g(X_{jn})|_\infty \big\}^3.
\end{eqnarray*}
Therefore, $ n^{-2} \sum_{i, j}^n |k_{ij}|  | (v_{ij}^* - l_{ij}^{(2)})_{st} |$ is bounded above by 
\[   \big(8d L |\beta_n^* - \beta_n|_\infty \big)  n^{-2} \sum_{i, j = 1}^n  |k_{ij}|  \big\{ |g(X_{in})|^3_\infty + |g(X_{jn})|^3_\infty \big\}.  \]
Now, by Condition~7(c), 
\[ n^{-2}  \sum_{i, j = 1}^n  \E \left [  |k_{ij}|  \big\{ |g(X_{in})|^3_\infty + |g(X_{jn})|^3_\infty \big\} \right] = O(1),\]
 and hence \eqref{eq:remaninder_first_term} follows. We can apply  similar techniques to control the other two terms in $Q_n$. Hence, $Q_n = o_p(1)$. 
\subsubsection{Finding the limiting distribution} \label{subset:limit_dist}
In this subsection, we will finally prove that $nT_n^*$ converges to a non-degenerate distribution. By \eqref{eq:T_n*}, \eqref{eq:triangular_removed} and $\eqref{eq:Rnto0}$, it is enough to show that the random variable
\[ n V_n^{(0)}  +   \zeta_n^\T  (n^{1/2} V_{n}^{(1)} )  +   \frac{1}{2}  \zeta_n^\T V_{n}^{(2)}  \zeta_n \]
converges in distribution, where $V_n^{(p)} (p = 0, 1, 2)$ is defined near \eqref{eq:neg_TtoV}. The kernel  $h^{(0)}$ is degenerate of order $1$, i.e., $\E \{h^{(0)}(z_1, Z_2, Z_3, Z_4) \} = 0 $ almost surely.  Define 
\[ h^{(0)}_2(z_1, z_2) =  \E \{ h^{(0)}(z_1, z_2, Z_3, Z_4) \} \]
and let $S_n^{(0)}$ be the V-statistic with kernel $ h^{(0)}_2$, i.e.,
\[ S_n^{(0)}  = n^{-2} \sum_{i, j=1}^n h^{(0)}_2(Z_i, Z_j).  \]
By the standard theory of V-statistics, 
\[ n(V_n^{(0)} - S_n^{(0)}) \to 0,\]
in probability. 
The symmetric function $h^{(0)}_2$ admits an eigenvalue decomposition 
\[ h^{(0)}_2(z_1, z_2) = \sum_{r = 0}^\infty \lambda_r \varphi_r(z_1) \varphi_r(z_2)\]
where $(\varphi_r)_{ r \ge 0}$ is an orthonormal basis of $L_2(\mathbb R^{d_0+1}, P)$ and  $\lambda_r$ is the eigenvalue corresponding to the eigenfunction $\varphi_r$. Since $h^{(0)}_2$ is degenerate of order $1$,  $\lambda_0 = 0,  \varphi_0 \equiv 1$.   Therefore, $\E\{\varphi_r(Z_1)\} = 0$ for each $r \ge 1$. Also, $\sum_{ r} \lambda_r^2  = \E \{ h^{(0)}_2(Z_1, Z_2)^2\} < \infty$.  We use the above decomposition of $h^{(0)}_2$ to express $S_n^{(0)}$  as 
\[ S_n^{(0)} =  \sum_{ r=1}^\infty \lambda_r \Big\{ n^{-1/2} \sum_{i=1}^n \varphi_r(Z_i) \Big\}^2.\]
Let us now turn our attention to $V_n^{(1)}$. It is again a V-statistic whose  kernel $h^{(1)}$ has mean zero, i.e., $\E \{ h^{(1)}(Z_1, Z_2, Z_3, Z_4)\} = 0$. See Lemma~\ref{l:kernel_h1_mean_0} below for a proof.  Therefore, if we define its first order projection by 
\[ h_1^{(1)}(z_1) =  \E \{ h^{(1)}(z_1, Z_2, Z_3, Z_4)\}, \]
then 
\[  n^{1/2} V_n^{(1)} -  n^{-1/2} \sum_{i=1}^n h_1^{(1)}(Z_i)  \to 0,\]
in probability.  On the other hand, by the weak law of large numbers for V-statistics, we have 
\[  V_n^{(2)}  \to \E \{ h^{(2)}(Z_1, Z_2, Z_3, Z_4)\} = \Lambda \in \mathbb R^{d^2},\]
in probability. 
 By the multivariate central limit theorem, the random vectors 
 \[  \Big\{ n^{-1/2} \sum_{i=1}^n \varphi_r(Z_i)\Big\}_{ r \ge 1},  \  n^{-1/2} \sum_{i=1}^n h_1^{(1)}(Z_i), \  \zeta_n, \] converge in distribution  to  jointly Gaussian random variables
 \[ \mathcal{Z} = (\mathcal{Z}_r)_{ r \ge 1}, \mathcal{N}= (\mathcal{N}_i)_{1 \le i \le d}, \mathcal{W} = (\mathcal{W}_i)_{1 \le i \le d},\]
where $\mathcal{Z}_r$ are independent and identically distributed  $N(0,1)$, and  the random vectors $\mathcal{N}$ and  $\mathcal{W}$ are distributed as  $N_d(0, \Xi)$ and $N_d\big(0,  \alpha I\big)$ respectively with $\alpha =  \E (\eps_1^2) $ and $ \Xi = \E \{ h_1^{(1)}(Z_1) h_1^{(1)}(Z_1)^\T\}$.  Also, the covariance structure between the random variables  $\mathcal{Z}_r, \mathcal{N} $ and $\mathcal{W}$ are given by 
\begin{align*} \E(\mathcal{Z}_r \mathcal{N} ) = \E\{\varphi_r(Z_1)h_1^{(1)}(Z_1)\},& \quad
\E ( \mathcal{Z}_r  \mathcal{W}) =  A^{-1}\E \{g(X_1) \eps_1 \varphi_r(Z_1)\},\\
 \E (\mathcal{W} \mathcal{N}^\T ) &=  A^{-1}\E \{\eps_1 g(X_1)  h_1^{(1)}(Z_1)^\T\}.
\end{align*}
Therefore, by the continuous mapping theorem,
\begin{eqnarray}
n T_n^* & = & n V_n^{(0)}  +   \zeta_n^\T  (n^{1/2} V_{n}^{(1)} )  +   \frac{1}{2}  \zeta_n^\T V_{n}^{(2)}  \zeta_n + o_p(1)  \nonumber \\
	& \to & \sum_{ r=1}^\infty \lambda_r \mathcal{Z}_r^2+ \sum_{ i=1} ^d \mathcal{W}_i  \mathcal{N}_i + \frac12  \sum_{ i, j=1} ^d \mathcal{W}_i  \Lambda_{ij} \mathcal{W}_j = \chi, \label{eq:LimDist}
\end{eqnarray}
in distribution, which concludes the proof of the theorem.

\begin{lemma} \label{l:kernel_h1_mean_0}
Let $h^{(1)}$ be the symmetric kernel as defined in \eqref{def:hp}. Let $Z_1, Z_2, Z_3$ and $Z_4$ be independent and identically distributed random vectors with 
$Z_i = (X_i, \eps_i) \in \mathbb R^{d_0} \times \mathbb R$ where $X_i$ and $\eps_i$ are independent.  Then 
\[ \E\{h^{(1)}(Z_{1}, \ldots, Z_{4})\} = 0 .\]
\end{lemma}
\begin{proof}
We have
\[ h^{(1)}(Z_{1}, Z_{2},  Z_{3}, Z_{4}) = \frac{1}{4!} \sum_{(t, u, v, w)}^{(1, 2, 3, 4)}  k_{tu} l_{tu}^{(1)} + k_{tu} l_{vw}^{(1)}  - 2 k_{tu}l_{tv}^{(1)}, \]
where the sum is  over all $4!$ permutations of $(1,2,3,4)$.  
The lemma would follow immediately if
$\E (k_{tu} l_{tu}^{(1)} + k_{tu} l_{vw}^{(1)}  -  2k_{tu}l_{tv}^{(1)} ) = 0$ for each such permutation. Recall that
\[ l_{ij}^{(1)}  =   -  \Big \{  l_x(\eps_{i}, \eps_{j}) g(X_{i}) +  l_y(\eps_{i}, \eps_{j}) g(X_{j}) \Big \}. \]
Express the  right hand side of the above equation as $Q_{ij} + R_{ij}$.  Using the independence of $X_i$ and $\eps_{i}$, 
\[ \begin{split}
& \E ( k_{tu} Q_{tu} + k_{tu} Q_{vw}  -  2k_{tu}Q_{tv} )  \\
=&  - \E \{ k(X_{t}, X_{u})g(X_{t})\}\E \{l_x( \eps_{t}, \eps_{u})\} -  \E\{k(X_{t}, X_{u})g(X_{v})\}\E\{l_x( \eps_{v}, \eps_{w})\} \\
& \hspace{1in} +  2 \ \E \{k(X_{t}, X_{u})g(X_{t})\}\E \{ l_x( \eps_{t}, \eps_{v})\}  \\
=& \ \E \{l_x( \eps_{1}, \eps_{2})\}  \big [\E \{k(X_{1}, X_{2})g(X_{1}) \} - \E \{k(X_{1}, X_{2})g(X_{3})\}    \big].
\end{split}\]
Similarly, 
\[ \begin{split}
&  \E (k_{tu} R_{tu} + k_{tu} R_{vw}  -  2k_{tu}R_{tv})  \\
=&  - \E \{k(X_{t}, X_{u})g(X_{u})\}\E\{l_y( \eps_{t}, \eps_{u})\} -  \E \{ k(X_{t}, X_{u})g(X_{w})\}\E\{l_y( \eps_{v}, \eps_{w})\} \\
&  \hspace{1in}  +  2\ \E\{k(X_{t}, X_{u})g(X_{v})\}\E\{l_y( \eps_{t}, \eps_{v})\}  \\
=& \  \E \{l_y( \eps_{1}, \eps_{2})\}  \big[\E\{k(X_{1}, X_{2})g(X_{3})\} - \E\{k(X_{1}, X_{2})g(X_{2})\}    \big].
\end{split}\]
Since $k$ is symmetric, $\E \{k(X_{1}, X_{2})g(X_{2})\}  = \E \{ k(X_{1}, X_{2})g(X_{1})\} $ and since $l$ is symmetric,  $l_x(a, b)  = l_y(b, a)$ which implies that $ \E \{l_x( \eps_{1}, \eps_{2})\} = \E \{l_y( \eps_{1}, \eps_{2})\} $. Hence,
\[  \E (k_{tu} Q_{tu} + k_{tu} Q_{vw}  -  2k_{tu}Q_{tv} ) + \E ( k_{tu} R_{tu} + k_{tu} R_{vw}  -  2k_{tu}R_{tv} ) = 0,\]
and consequently, $\E (k_{tu} l_{tu}^{(1)} + k_{tu} l_{vw}^{(1)}  -  2k_{tu}l_{tv}^{(1)} ) =0$. 
\end{proof}

%\appendix

\subsection{The empirical distribution of the residuals}
In the following lemma we gather a few standard results about the  empirical distribution of the residuals for the linear regression model $Y = m(X) + \eta$.\begin{lemma}\label{lem:Cons_epsDist}
Under Conditions~1 and 5(a)--(b), the following statements hold:

{\em (i)} for  each  $0 < r \le 4+ 2 \delta, $ $\mathbb{P}_n  ( |e^o - \eps^o|^r) \to 0$,   almost surely; 

{\em (ii)} almost surely,  $P_{n,e^o} \rightarrow \eps^o,$ in distribution;

{\em (iii)} almost surely, $\sup_{n}  \prob_n ( |e^o|^{2 +\delta})  < \infty$. 
 \end{lemma}
\begin{proof}
Write $e_i - \eps_i =  - g(X_i) ^\T (\hat \beta_n  - \tilde \beta_0)$. Thus, 
\[\mathbb{P}_n  (|e - \epsilon|^r )  \le   d |\hat \beta_n - \tilde \beta_0|^r_{\infty}    \mathbb{P}_n \{ |g(X)|_\infty^r\}.\]
Hence, almost surely, $\mathbb{P}_n  (|e - \epsilon|^r ) \to 0$ using the facts that $\E  \{ |g(X)|^{4+ 2 \delta}_\infty\} < \infty$  by Condition~5(a)  and that $ \hat \beta_n  \to \tilde \beta_0 $ almost surely,  by \eqref{beta_n_expansionCLT} and $\E \{ | g(X) \eps|\} < \infty$, the latter being guaranteed by Conditions~5(a)--(b). Therefore,  almost surely, 
\begin{equation*}\label{eq:e_bar_mu}
\;\;\;\;\;\;\;\; \bar{e} = \mathbb{P}_n (e) = \mathbb{P}_n \{m(X)\} -  \mathbb{P}_n\{g(X)\}^\T   \hat \beta_n \to \E \{m(X)\} -  \E\{g(X)\}^\T  \tilde \beta_{0}= \E(\eps). 
\end{equation*}
This completes the proof of (i).

Let $P_{n,\eps^o}$  be the empirical measure of $\epsilon_1^o, \ldots,\epsilon_n^o$. Its characteristic function is
\[ \int e^{i\xi x}d P_{n,e^o}(x) = e^{-i\xi\bar{e}} \mathbb{P}_n (e^{i\xi {e}}).\]
Hence, by applying part (i) of the lemma with $r = 1$, for any $\xi\in\mathbb{R}$, 
\begin{eqnarray}
 \left|\int e^{i\xi x}d P_{n,e^o} (x) -  e^{-i\xi \{\bar{e} -  \E(\eps)\}}\int e^{i\xi x}d P_{n,\eps^o}(x) \right| &=& \left|\mathbb{P}_n ({e^{i\xi {e}}}) -\mathbb{P}_n (e^{i\xi\epsilon})\right| \nonumber\\
 &\leq& |\xi|\ \mathbb{P}_n(|e - \epsilon|) \to 0,\nonumber
\end{eqnarray}
almost surely. Now by the Glivenko-Cantelli lemma  almost surely, $P_{n,\eps^o} \to \eps^o$  in distribution. Next   $\bar e \to \E (\eps)$, again almost surely,  as shown in   part~(i) of the lemma. Therefore,   $ \int e^{i\xi x}d P_{n,e^o}(x) \to \int e^{i\xi x}d P_{\eps^o} (x)$, almost surely,  which, by the L\'evy's continuity theorem,  yields (ii).

\noindent To prove (iii), we write 
\begin{eqnarray*}
\prob_n ( |e^o|^{2 +\delta}) & = & \mathbb{P}_n (|e - \bar{e}|^{2 + \delta}) = \mathbb{P}_n \big[|(e - \eps) + \{\eps -\E(\eps)\} - \{\bar{e} - \E(\eps)\}|^{2 + \delta}\big] \\
	& \le& 3^{2+\delta} \Big\{  \mathbb{P}_n (|e - \eps|^{2 + \delta}) + \mathbb{P}_n (|\eps^o|^{2 + \delta})  + |\bar{e} - \E(\eps)|^{2 + \delta} \Big\}.
\end{eqnarray*}
The result is then an immediate consequence of the  fact  that $ \mathbb{P}_n  (|\eps^o|^{2 + \delta})  \to \E ( |\epsilon^o|^{2 + \delta}) <\infty$ almost surely  by Conditions~5(a)--(b), that  $\bar e \to \E (\eps)$ almost surely,  and part (i) of the lemma. 
\end{proof}

\section{Acknowledgment} The authors thank Probal Chaudhuri, Victor de la Pe{\~n}a, Bharath Sriperumbudur and  G\'{a}bor Sz\'{e}kely for helpful discussions. We thank the reviewers for their helpful comments.

\bibliography{cond_ind}

\end{document}